\documentclass[journal,twocolumn]{IEEEtran}
\usepackage{graphicx}
\usepackage{dcolumn}
\usepackage{bm}
\usepackage{braket}
\usepackage{footmisc}
\usepackage{amsmath}
\usepackage{amsfonts}
\usepackage[normalem]{ulem}
\usepackage{mathtools}
\usepackage{cancel}
\usepackage{soul}
\usepackage{cite}
\usepackage{algorithm, algpseudocode}
\usepackage{algorithmicx}
\usepackage{amsthm}
\usepackage{cite}
\newtheorem{proposition}{Proposition}

\newtheorem{lemma}{Lemma}
\newtheorem{definition}{Definition}
\usepackage{comment}
\usepackage{xcolor}

\begin{document}
%
\title{Energy-constrained LOCC-assisted quantum capacity of bosonic dephasing channel }
%
%
%

\author{Amir~Arqand,
        Laleh Memarzadeh,
        and~Stefano~Mancini
\thanks{A. Arqand and L. Memarzadeh are with the Department of Physics, Sharif University of Technology, Tehran, Iran,
P.O. Box: 11365-9161.}
\thanks{S. Mancini, School of Science \& Technology, University of Camerino, I-62032 Camerino, Italy and
INFN Sezione di Perugia, I-06123 Perugia, Italy }
}

\maketitle

\begin{abstract}
We study the LOCC-assisted quantum capacity of bosonic dephasing channel with energy constraint on input states. We start our analysis by focusing on the energy-constrained squashed entanglement of the channel, which {is an} upper bound for the energy-constrained LOCC-assisted quantum capacity. As computing energy-constrained squashed entanglement of the channel is challenging due to a double optimization (over the set of density matrices and the isometric extensions of a squashing channel), we first derive an upper bound for it, and then we discuss how tight that bound is for energy-constrained LOCC-assisted quantum capacity of bosonic dephasing channel. We prove that the optimal input state is diagonal in the Fock basis. 
Furthermore, {we prove that for a generic channel, the optimal squashing channel
belongs to the set of symmetric quantum Markov chain inducer (SQMCI) channels of the channel system-environment output, provided that such a set is non-empty.
With supporting arguments, we conjecture that this is instead the case for the bosonic dephasing channel. Hence, for it we analyze two explicit examples of squashing channels which are not SQMCI, but are symmetric. Through them}, we derive explicit upper and lower bounds for the energy-constrained LOCC-assisted quantum capacity of the bosonic dephasing channel in terms of its quantum capacity with different noise parameters. {As the difference between upper and lower bounds is at most of the order $10^{-1}$, we conclude that the bounds are tight. Hence we} provide a very good estimation of the LOCC-assisted quantum capacity of the bosonic dephasing channel. 
\end{abstract}


%
\IEEEpeerreviewmaketitle

\section{Introduction}
One of the essential steps for the implementation of quantum information protocols and the development of quantum technology is the establishment of reliable communication between two parties. That motivates analyzing the capacity of quantum channels, especially quantum capacity, which corresponds to the reliable rate of sharing entanglement between two points. 

As continuous-variable systems are promising candidates for quantum communication, analyzing the capacity of channels defined over infinite-dimensional Hilbert-spaces is of practical and theoretical importance. In this set of channels, the subset of Gaussian channels that maps Gaussian states to Gaussian states has been studied extensively \cite{Gchannel,Gchannel2,Holevo2001,Harrington2001,Giovannetti2004,Caves2004}. However, there is a strong motivation to go beyond Gaussian channels for having better performance in tasks such as parameter estimation \cite{Adesso2009} and teleportation \cite{Opatrny2000,Mista2006,Olivares2003} 
or to bypass the limitations of Gaussian maps for entanglement distillation \cite{Eisert2002,Fiurasek2002, Fiurasek2002-2,Giedke2002}, error correction \cite{Niset2009}, and quantum repeaters \cite{Namiki2014}. 

In general, computing quantum capacity is challenging because of two necessities. First, {the} optimization of coherent information over the set of input density operators, and second its regularization \cite{Devetak2005}.
The situation {becomes more complicated for non-Gaussian channels compared to Gaussian ones because one cannot limit the analysis to states of Gaussian form, characterized just by covariance matrix and displacement vector.} That makes obtaining analytical {or numerical results for non-Gaussian channels a daunting task.} 
Despite such technical difficulties, {recently, there has been an increasing attention to non-Gaussian channels \cite{Memarzadeh2016,Sabapathy2017,qc_dephasing,Lami2020S}. In particular, in \cite{qc_dephasing} it has been shown that the quantum capacity of bosonic dephasing channel, as an example of a non-Gaussian channel, is achieved by a Gaussian mixture of Fock states.}  

{Here, we are interested in finding the energy-constrained LOCC-assisted quantum capacity of the bosonic dephasing channel. It is worth mentioning that the bosonic dephasing channel describes a snapshot of a quantum Markov process and the channel noise parameter is proportional to the time the bosonic system interacts with the environment in the weak-coupling limit \cite{qc_dephasing, jiang}. 
Furthermore, dephasing is an unavoidable source of noise
in photonic communications \cite{gordon}. 
This happens, for instance,
with uncertainty path length in optical fibers \cite{Derickson1998FiberOT}.

{Energy-constrained LOCC-assisted quantum capacity of bosonic dephasing channel} is the maximum rate at which entanglement can reliably be established between sender and receiver 
when local operations 
and classical communication (LOCC) between sender and receiver is also allowed.
Additionally, we consider energy constraint on the channel input. {The importance of sharing entanglement is related to the key role of this correlation in the implementation of quantum protocols. This is not limited to theoretical investigations and is actually the cornerstone of developing quantum networks \cite{kimble}. This motivates analyzing any factor that affects the rate of entanglement sharing, including LOCC between the sender and the receiver.}

{It was shown that the  LOCC-assisted quantum capacity} is upper bounded by energy-constrained squashed entanglement of the channel \cite{wildelocc}. Computing this upper bound is another challenge because it requires two optimizations, one over the set of input density operators and {another over the set of isometric extensions of a squashing channel.} 

In order to facilitate performing the optimization for computing {the channel energy-constrained squashed-entanglement, we shall use the channel symmetry and analytical techniques to restrict the search over smaller sets of density operators and isometric extensions.} Analytically, we {will} characterize the subset of density operators that includes the optimal input states.
Furthermore, {for a generic channel, we prove that the
the set of SQMCI channels acting on the environment output of the original channel 
must (if not null) contain the optimal squashing channel.
For the bosonic dephasing channel, we conjecture that this set is null, and hence we restrict our analysis to symmetric squashing channels.} For two cases of {50/50 beamsplitter}. {We will analytically prove that for 50/50 beamsplitter} squashing channel, there is an upper bound and a lower bound for LOCC-assisted quantum capacity of bosonic dephasing channel 
with/without energy constraint, in terms of its quantum capacity with/without energy constraint. Numerically we {shall} compute these bounds for inputs subject to energy constraint, which will result in tight bounds. We shall also discuss the value of these bounds when there is no input energy constraint. We also study {symmetric} qubit squashing channels {and in this subset numerical search for the optimal squashing channel.}

The structure of the paper is as follows. In \S\ref{sec:back_not} we {set} our notation and provide essential background on squashed entanglement, LOCC-assisted quantum capacity, and degradability of quantum channels. {\rm Here, we also recall the bosonic dephasing channel.} In \S\ref{sec:optimal_input} we introduce the structure of optimal input state. In {\S\ref{sec:squashing} we define quantum Markov chain inducer (QMCI) and symmetric quantum Markov chain inducer (SQMCI) channels and show how optimal squashing channels and SQMCI channels are related. \S\ref{sec:examples} is devoted to two explicit examples for squashing channels for the bosonic dephasing channel: 50/50 beamsplitter and symmetric qubit channels}
 We summarise and discuss the results in \S\ref{sec:conc}.

\section{Background and notation}\label{sec:back_not}
In this section, we {set} our notation and provide the background required {to follow} the discussions in the next sections.
\subsection{Notation}
In this subsection, we set our notation. {Throughout the paper  we shall mainly deal with 
four input (output) systems. 
 ``$S$'' and ``$S'$'' label respectively the input and the output main system.
Similarly, ``$E$'' and ``$E'$'' label respectively the input and the output environment.
``$R$'' labels the reference system that remains unaltered from input to output.
Finally, ``$F$'' and ``$F'$'' denotes the input and the output environment for squashing channel that we shall introduce later on. The associated Hilbert spaces will be denoted by $\mathcal{H}_X$ and 
$\mathcal{H}_{X'}$ where $X$ can be either $R, S, E, F$ or combinations of them. }
 
By $\mathcal{N}_{X\rightarrow X'}$, we denote a completely positive trace preserving (CPTP) map or, for short, a quantum channel:
{\begin{equation}
    \mathcal{N}_{X\rightarrow X'}: \mathcal{T}(\mathcal{H}_X)\rightarrow\mathcal{T}(\mathcal{H}_{X'}),
\end{equation}}
where $\mathcal{T}(\mathcal{H}_X)$ stands for the set of trace-class operators on $\mathcal{H}_X$. Furthermore, by $\mathcal{L}(\mathcal{H}_X)$ we represent the set of linear operators on Hilbert space $\mathcal{H}_X$.

A unitary extension of channel  $\mathcal{N}_{X\rightarrow X'}$, is a unitary operator
$U:\mathcal{H}_X\otimes\mathcal{H}_Y\to \mathcal{H}_{X'}\otimes\mathcal{H}_{Y'}$ where  $\mathcal{H}_X\otimes\mathcal{H}_Y$ is isomorphic with $\mathcal{H}_{X'}\otimes\mathcal{H}_{Y'}$ , such that:
\begin{align}
    \mathcal{N}_{X\rightarrow X'}(\rho_X)&={\rm Tr}_{Y'}\left(\mathcal{U}^{\mathcal{N}}_{XY}(\rho_X\otimes \ket{0}\bra{0}_Y)\right)\cr
    &={\rm Tr}_{Y'}(U(\rho_X\otimes \ket{0}\bra{0}_Y)U^\dagger)
\end{align}

for all $\rho_X\in \mathcal{T}(\mathcal{H}_X)$, where
\begin{equation}\label{eq:uni_ext}
    \mathcal{U}^{\mathcal{N}}_{XY}[\bullet]:=U\bullet U^\dagger.
\end{equation} 
Similarly, an isometric extension of channel $\mathcal{N}_{X\rightarrow X'}$ is an isometry 
$V: \mathcal{H}_X \to \mathcal{H}_{X'}\otimes\mathcal{H}_{Y'}$, such that:
\begin{equation}
    \mathcal{N}_{X\rightarrow X'}(\rho_X)={\rm Tr}_{Y'}\left(\mathcal{V}^{\mathcal{N}}_{X\rightarrow X'Y'}(\rho_X)\right)={\rm Tr}_{Y'}(V\rho_XV^\dagger),
\end{equation}
for every $\rho_X\in\mathcal{T}(\mathcal{H}_X)$, where
 \begin{equation}\label{eq:iso_extension}
     \mathcal{V}^{\mathcal{N}}_{X\rightarrow X'Y'}[\bullet]:=V\bullet V^\dagger.
 \end{equation}
 Purification of density matrix $\rho_X$ is denoted by $\ket{\phi_{XY}}$ and the density {operator corresponding to it is}
 \begin{equation}
 \label{eq:StatePurification}
     \phi_{XY}:=\ket{\phi_{XY}}\bra{\phi_{XY}}.
 \end{equation}
The von Neumann entropy of an arbitrary state $\rho$ {is} 
\begin{equation}
    \label{eq:von_neumann entropy}
    \mathcal{S}(\rho):=-\operatorname{Tr}(\rho\log\rho).
\end{equation}
Throughout the paper, we use the logarithm to base two. We recall that the conditional entropy of a bipartite quantum state $\rho_{XY}$ is defined as follows:
\begin{equation}
\label{eq:conditional entropy}
    \mathcal{S}(X|Y)_{\rho_{XY}}:=\mathcal{S}(\rho_{XY})-\mathcal{S}(\rho_Y),
\end{equation}
where $\rho_Y=\operatorname{Tr}_Y(\rho_{XY}$).
For a bipartite quantum state $\rho_{SS'}\in\mathcal{T}(\mathcal{H}_S\otimes\mathcal{H}_{S'})$, the mutual information { ${I}(S;S')_{\rho_{SS'}}$} quantifies the correlation between subsystems with reduced density matrices $\rho_S=\operatorname{Tr}_{S'}(\rho_{SS'})$ and $\rho_{S'}=\operatorname{Tr}_{S}(\rho_{SS'})$. It is defined as:
\begin{equation}
\label{eq:mutual_informatio}
{I}(S;S')_{\rho_{SS'}}{:=} \mathcal{S}(\rho_S)+\mathcal{S}(\rho_{S'})-\mathcal{S}(\rho_{SS'}).
\end{equation}

Moreover, for a tri-partite quantum state $\rho_{SS'R}\in\mathcal{T}(\mathcal{H}_S\otimes\mathcal{H}_{S'}\otimes\mathcal{H}_R)$, conditional mutual information ${I}(S;S'|R)_{\rho_{SS'R}}$ quantifies the correlation between density matrices of subsystems $\rho_{SS'}=\rm{Tr}_R(\rho_{SS'R})$ and $\rho_R=\rm{Tr}_{SS'}(\rho_{SS'R})$. This positive quantity is given by
    \begin{equation}
        \label{eq:CondMut}
        {I(S; S'|R)}_{\rho_{SS'R}}:=\mathcal{S}(S|R)_{\rho_{SR}}+\mathcal{S}(S'|R)_{\rho_{S'R}}-\mathcal{S}({SS'}|R)_{\rho_{SS'R}},
    \end{equation}
    where the conditional entropy of a bipartite state is defined in Eq.~(\ref{eq:conditional entropy}).
\subsection{\label{sec:SQ_E}Squashed Entanglement}
In this subsection, we review the definition of {quantities necessary} for introducing the upper bound on the two-way LOCC-assisted quantum capacity of a channel. First, we recall the definition of squashed entanglement of a bipartite system. Then we proceed with reviewing the definition of {squashed entanglement of a channel} and energy constraint squashed entanglement of a channel. 

Squashed entanglement is an entanglement monotone for bipartite quantum states \cite{CW2004} which is a lower bound for the entanglement of formation and an upper bound on the distillable entanglement. It is defined as follows:
\begin{definition}
\label{Def:SquashedEntanglement}
    The squashed entanglement of a bipartite quantum state $\rho_{SS'}\in\mathcal{T}(\mathcal{H}_S\otimes\mathcal{H}_{S'})$ is defined as \cite{CW2004} 
    \begin{equation}\label{eq:ssq}
    E_{sq}(S;S')_{\rho_{SS'}}\coloneqq\frac{1}{2}\inf_{\rho_{SS'E'}} {I}
    (S;S'|{E'})_{\rho_{SS'E'}},
\end{equation}
where 
the infimum is taken over all extensions of $\rho_{SS'}$, that is over all quantum states $\rho_{SS'E'}$ such that $\rho_{SS'}={\rm Tr}_{E'}(\rho_{SS'E'})$. 
\end{definition}

Using the concept of bipartite state squashed entanglement, the {squashed entanglement of a channel} was introduced in \cite{TGW2014}. 
It represents the maximum squashed entanglement that can be generated by the channel.
\begin{definition}
   The squashed entanglement of a channel $\mathcal{N}_{S\rightarrow S'}$, is given by \cite{TGW2014}:
    \begin{equation}\label{eq:csq}
    \tilde{E}_{sq}(\mathcal{N}_{S\rightarrow S'})=\sup_{\ket{\phi_{RS}}}E_{sq}(R;{S'})_{\rho_{RS'}},
\end{equation}
where the supremum is taken over all bipartite pure states $\ket{\phi_{RS}}\in\mathcal{H}_R\otimes\mathcal{H}_S$ and $\rho_{RS'}=(\operatorname{id}_R\otimes\mathcal{N}_{S\rightarrow S'})(\ket{\phi_{RS}}\bra{\phi_{RS}})$.

\end{definition}
An alternative, more practical expression for {squashed entanglement of a channel} is given by \cite{TGW2014}:
\begin{equation}
\label{eq:channel_squashing_alt}
    \tilde{E}_{sq}(\mathcal{N}_{S\rightarrow S'})=\sup_{\rho_S}
        E_{sq}(\rho_S,\mathcal{N}_{S\rightarrow S'}),
\end{equation}
where {supremum is over all input density operators, $\rho_S\in\mathcal{T}(\mathcal{H}_S)$, and}
\begin{align}
    \label{eq:not_for_inf}
    E_{sq}&(\rho_S,\mathcal{N}_{S\rightarrow S'}){:=}\frac{1}{2}\times\cr
    &\inf_{\mathcal{V}^{\mathcal{N}^{sq}_{E\rightarrow E'}}_{E\rightarrow E'F'}}\Big(\mathcal{S}(S'|E')_{{\sigma_{S'E'}}}+\mathcal{S}({S'}|{F'})_{{\sigma_{S'F'}}}\Big).
\end{align}
{Here} the infimum is taken over all isometric extensions of the squashing channel
{{(see Fig.~(\ref{fig:dilation_maps})) and $\sigma_{S'E'}$ and $\sigma_{S'F'}$ are respectively obtained by partial trace over degrees of freedom in $\mathcal{H}_{F'}$ and $\mathcal{H}_{E'}$ of the state 
}} 

\begin{figure}
    \includegraphics[width=\columnwidth]{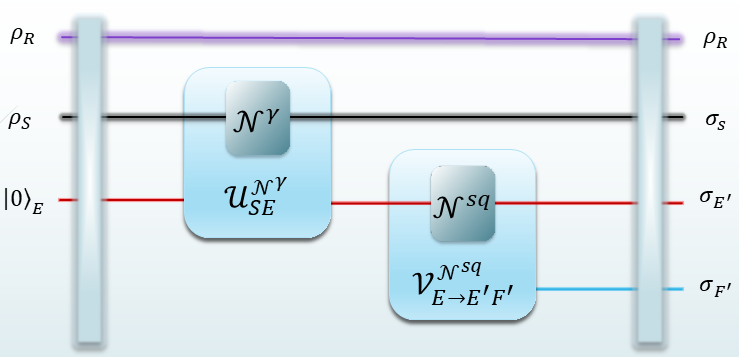}
    \caption{
    Schematic representation of the bosonic dephasing channel, squashing channel, and their isometric {extensions.} 
    }
    \label{fig:dilation_maps}
\end{figure}
\begin{equation}\label{eq:ext}
    \sigma_{S'E'F'}{:=}\left(\mathcal{V}_{E\rightarrow E'F'}^{\mathcal{N}^{sq}_{E\rightarrow E'}}\circ \mathcal{V}^{\mathcal{N}_{S\rightarrow S'}}_{S\rightarrow S'E}\right)(\rho_S),
\end{equation}
where $\mathcal{V}^{\mathcal{N}_{S\rightarrow S'}}_{S\rightarrow S'E}$ {and $\mathcal{V}_{E\rightarrow E'F'}^{\mathcal{N}^{sq}_{E\rightarrow E'}}$ are respectively} the conjugation of isometric
extension of the 
channel $\mathcal{N}_{S\rightarrow S'}$ {and $\mathcal{N}^{sq}_{E\rightarrow E'}$ (see Eq.~(\ref{eq:iso_extension}))}. The superscript $sq$ in $\mathcal{N}^{sq}_{E\rightarrow E'}$ is to label the squashing channel. If there exists a channel for which the infimum in Eq.~(\ref{eq:not_for_inf}) is achieved, we called it the optimal squashing channel.

The definition of the squashed entanglement of a channel can be generalized {to the case where there is a constraint on the energy of input states.}
\begin{definition}
\label{def:ECSEC}
For channel $\mathcal{N}_{S\rightarrow S'}$ with energy constraint at the input, that is ${\rm Tr}(\rho_S G)\leq N$ where $\rho_S$ represents an arbitrary input state, $G$ is the energy observable, and $N\in[0,\infty)$, the energy-constrained squashed entanglement of the channel is given by 
\begin{equation}\label{eq:ecsq}
    \tilde{E}_{sq}(\mathcal{N},G,N)
    =\sup_{\rho_S:{\rm Tr}(G\rho_S)\leq N}E_{sq}(\rho_S,\mathcal{N}_{S\rightarrow S'}),
\end{equation}
where, $E_{sq}(\rho_S,\mathcal{N}_{S\rightarrow S'})$ is defined in Eq.~(\ref{eq:not_for_inf}).
    
\end{definition}

\subsection{Two-way LOCC-assisted quantum capacity}
In this subsection, we bring {to light} the definition of two-way LOCC-assisted quantum capacity, and its energy-constrained form \cite{TGW2014,wildelocc}. Then we recall its upper bound in terms of the {squashed entanglement of a channel}. 

{The performance of quantum channels for reliable quantum communication is quantified by quantum capacity when there is no extra resources, such as shared entanglement or classical communication between sender and receiver. Indeed by allowing further resources, we expect higher rates of information transmission through the channel.
When LOCC is allowed interactively  between sender and receiver, the capability of channel for quantum communication is quantified by its two-way LOCC assisted quantum capacity, which is defined as follows:}
\begin{definition}\label{def:2wayQ}
    The two-way LOCC-assisted quantum capacity $Q^{LOCC}_{S\leftrightarrow S'}(\mathcal{N_{S\rightarrow S'}})$ of quantum channel $\mathcal{N_{S\rightarrow S'}}$ is the highest achievable rate of faithful qubit transmission {(through infinitely many uses)} of the channel with the assistance of unlimited two-way classical 
    communication \cite{BBPSCHW1996,BSW1996}.
\end{definition}
{The above definition is generalized for the situations where there is an upper bound on the average input energy:}
\begin{definition}
  {The energy-constrained} two-way LOCC-assisted quantum capacity $Q^{LOCC}_{S\leftrightarrow S'}(\mathcal{N}_{{S\rightarrow S'}},G,N)$ of a quantum channel $\mathcal{N}_{S\rightarrow S'}$ is {the} two-way LOCC-assisted quantum capacity {of Definition \ref{def:2wayQ},} with the constraint that the {average input energy per channel use {(denoted by observable $G$)} is not larger than $N$.}  \end{definition}
{Note that we could have considered a uniform energy constraint at the input (constraining the energy of each input) instead of considering the average input energy constraint (see e.g. \cite{wildelocc}). However, the two-way LOCC-assisted quantum capacity with uniform energy constraint is upper bounded by the two-way LOCC-assisted quantum capacity with average energy constraint on input. As we derive an upper bound on the latter, the results provide an upper bound for the former too.

Despite the importance of two-way LOCC assisted quantum capacity, there is no explicit compact expression to compute this capacity for a given channel. However,
a}ccording to \cite{wildelocc}, an upper bound on $Q^{LOCC}_{S\leftrightarrow  S'}(\mathcal{N}_{S\rightarrow S'},G,N)$ is given by {squashed entanglement of the channel}: 
\begin{equation}
\label{eq:TLOCCAQEC}
    Q^{LOCC}_{S\leftrightarrow S'}(\mathcal{N}_{S\rightarrow S'},G,N)\leq \tilde{E}_{sq}(\mathcal{N}_{S\rightarrow S'}, G, N),
\end{equation}
where the right hand is given in Eq.~(\ref{eq:ecsq}).
\subsection{Degradable and anti-degradable channels}\label{subsec:deg_anti-deg}
Here we review the definitions of degradable and anti-degradable channels \cite{Cubitt} {The concept of degradable and anti-degradable channels have been playing an important role in various aspects of quantum Shannon theory, including the subject of {quantum channel capacity} \cite{Khatri2020}.

For defining degradable and anti-degradable channels, first the complementary channel needs to be introduced.}
For a channel 
\begin{align}
\label{eq:NforDAD}
       &\mathcal{N}_{X\rightarrow X'}:\bullet\mapsto  {\rm Tr}_{Y}\left(\mathcal{U}^{\mathcal{N}}_{XY}(\bullet\otimes\ket{0}\bra{0})\right),
    \end{align}
 with $\mathcal{U}^{\mathcal{N}}_{XY}$ defined in Eq.~(\ref{eq:uni_ext}), the complementary channel $\mathcal{N}^c_{{X\rightarrow Y'}}$ is given by
 \begin{align}
 \label{eq:Nc}
       &\mathcal{N}^c_{X\rightarrow Y'}:\bullet\mapsto  {\rm Tr}_{X}\left(\mathcal{U}^{\mathcal{N}}_{XY}(\bullet\otimes\ket{0}\bra{0})\right),
    \end{align}
 Setting the definition of complementary channel, degradability and anti-degradability of a channel {are} defined as follows: 
\begin{definition}
    \label{def:degradable}
    {The channel $\mathcal{N}_{{X\rightarrow X'}}$} in Eq.~({\ref{eq:NforDAD}}) is degradable if {there exists a CPTP map $\mathcal{D}_{X'\rightarrow Y'}^{\mathcal{N}}
    $ such that}
    \begin{equation}
        \label{eq:degrabale}
        \mathcal{N}^c_{{X\rightarrow Y'}}=\mathcal{D}^{\mathcal{N}}_{{X'\rightarrow Y'}}\circ\mathcal{N}_{{X\rightarrow X'}},
    \end{equation}
    where $\mathcal{N}^c_{{X\rightarrow Y'}}$ is complementary {to $\mathcal{N}_{{X\rightarrow X'}}$ according to Eq.~(\ref{eq:Nc}).}
\end{definition}
{In simple terms, if a channel is degradable it is possible to construct the environment final state from the output state of the channel by means of a CPTP map. For anti-degradable channels the construction of channel output state from environment final state is possible with a CPTP map. The precise definition follows:}
\begin{definition}
    \label{def:anti-degradable}
   {The channel $\mathcal{N}_{{X\rightarrow X'}}$} in Eq.~({\ref{eq:NforDAD}}) is anti-degradable if {there exists a CPTP map $\mathcal{D}_{Y'\rightarrow X'}^{\mathcal{N}^c}
   $ such that}
    \begin{equation}
        \label{eq:anti-degrabale}
        \mathcal{N}_{{X\rightarrow X'}}=\mathcal{D}_{{Y'\rightarrow X'}}^{\mathcal{N}^c}\circ\mathcal{N}_{{X\rightarrow Y'}}^c,
    \end{equation}
where $\mathcal{N}^c_{{X\rightarrow Y'}}$ is complementary to $\mathcal{N}_{{X\rightarrow X'}}$ {according to Eq.~(\ref{eq:Nc}).}
\end{definition}
Symmetric channels \cite{SmithSmolinWinter2008, MorganWinter2013} are those channels {for which} 
\begin{equation}
    \mathcal{N}^c_{{X\rightarrow Y'}}=\mathcal{N}_{{X\rightarrow X'}}.
\end{equation}
{Indeed for symmetric channels $\mathcal{T}(\mathcal{H}_{X'})$ and $\mathcal{T}(\mathcal{H}_{Y'})$ are isomorphic. Regarding the definition of degradable and anti-degradable channels, it is concluded that
} symmetric channels belong to the intersection of the sets of degradable and anti-degradable channels.

\subsection{\label{sec:Q_CH}Quantum dephasing channel}
{The continuous variable bosonic dephasing channel 
$\mathcal{N}_{{S\rightarrow S'}}^\gamma$,
can successfully model decoherence in many different setups \cite{WM}. As the input space $\mathcal{H}_S$ and output space $\mathcal{H}_{S'}$ are isomorphic, from now on is we denote the bosonic dephasing channel with $\mathcal{N}^\gamma_{S\rightarrow S}$ where $\gamma\in[0,+\infty)$ is related to the dephasing rate.
Bosonic dephasing channels are }described through the following operator-sum representation \cite{qc_dephasing}.
\begin{equation}\label{phase_damping}
    \mathcal{N}_{{S\rightarrow S}}^\gamma(\rho)=\sum_{j=0}^{\infty}K_j\rho K^{\dagger}_j,
\end{equation}
where the Kraus operators are given by
\begin{equation}
    K_j=e^{-\frac{1}{2}\gamma(a^\dagger a)^2}\frac{(-\operatorname{i}\sqrt{\gamma}a^\dagger a)^j}{\sqrt{j!}},
\end{equation}
with $a^\dagger,a$ being bosonic creation and annihilation operators on $\mathcal{H}_S$ and $\gamma\in[0,+\infty)$ is related to the dephasing rate.

The channel can be dilated into a single mode environment using the following unitary $U_{\mathcal{N}_{{S\rightarrow S}}^\gamma}\in\mathcal{L}(\mathcal{H_S}\otimes\mathcal{H}_{E})$
\begin{eqnarray}
\label{eq:unitary_d}
    U_{\mathcal{N}_{{S\rightarrow S}}^\gamma}=&&e^{-\operatorname{i}\sqrt{\gamma}(a^\dagger a)(b+b^\dagger)}\nonumber\\
    =&&e^{-\operatorname{i}\sqrt{\gamma}(a^\dagger a)b^\dagger}e^{-\operatorname{i}\sqrt{\gamma}(a^\dagger a)b}e^{-\frac{1}{2}\gamma(a^\dagger a)^2},
\end{eqnarray}
with $b^\dagger,b$ being bosonic creation and annihilation operators on $\mathcal{H}_E$. The unitary \eqref{eq:unitary_d} has the form of a controlled dephasing with the environment's mode acting as control.

It is not hard to see that the {channel $\mathcal{N}_{{S\rightarrow S}}^\gamma$ has phase covariant property {under the unitary operator}, that is 
\begin{equation}\label{con:cov}
    \mathcal{N}_{{S\rightarrow S}}^\gamma(U_\theta \rho U_{\theta}^{\dagger})=U_\theta \mathcal{N}_{{S\rightarrow S}}^\gamma(\rho) U_{\theta}^{\dagger},
\end{equation}
where unitary operator $U_{\theta}$ is given by
\begin{equation}
    \label{Eq:PhaseO}
    U_{\theta}=e^{\operatorname{i}(a^\dagger a)\theta}
    {\in\mathcal{L}(\mathcal{H}_S).}
\end{equation}
}
Moreover, the output of the complementary channel can be {written as}
\begin{align}
\label{eq:complementary output}
\mathcal{N}^{\gamma^c}_{{S\rightarrow E}}(\rho)&={\operatorname{Tr}}_{S}[U_{\mathcal{N}_{{S\rightarrow S}}^\gamma}(\rho\otimes\ket{0}\bra{0})U_{\mathcal{N}_{{S\rightarrow S}}^\gamma}^\dagger]\notag\\
&={\sum_{n=0}^\infty} \rho_{nn}\ket{-\operatorname{i}\sqrt{\lambda}n}\bra{-\operatorname{i}\sqrt{\lambda}n},
\end{align}
where $\ket{-\operatorname{i}\sqrt{\lambda}n}{\in\mathcal{H}_{E}}$ {is a coherent state} of amplitude {$\sqrt{\lambda}n$ with phase $-\operatorname{i}$ and, $\ket{0}$ is the vacuum state of environment}.
By the above relation and using Eq.~(\ref{Eq:PhaseO}) we can see that
\begin{equation}\label{eq:comp_inv}
    \mathcal{N}_{{S\rightarrow E}}^{\gamma^c}(U_\theta\rho U_\theta^\dagger)=\mathcal{N}_{{S\rightarrow E}}^{\gamma^c}(\rho).
\end{equation}
which means that the complementary channel of $\mathcal{N}_{{S\rightarrow E}}^\gamma$ is invariant under the {unitary \eqref{Eq:PhaseO}}.

\section{Optimal input state}\label{sec:optimal_input}
In this section, we derive an upper bound for the squashed entanglement of the bosonic dephasing channel defined in Eq.~(\ref{phase_damping}). In doing so, we prove that the optimal input state for which such an upper bound can be achieved is diagonal in the Fock basis. We use the structure of optimal input state {to simplify} the expression for {squashed entanglement of the channel} which we {will use in subsequent} sections. sections.

{To analyze energy-constrained squashed entanglement (see Def.~\ref{def:ECSEC}) for the bosonic dephasing channel as energy observable $G$ we use the operator $a^\dag a$ because for a bosonic mode it corresponds (up to a constant) to the Hamiltonian. }
\begin{proposition}\label{prop:fock_optimal}
For a bosonic dephasing channel with parameter $\gamma$ and energy observable $G=a^\dag a$, the supremum in Eq.~(\ref{eq:ecsq}) is achieved by diagonal states in the Fock basis. 
\end{proposition}

\begin{proof}
Define $U_{SX}:(\mathcal{H}_S\otimes\mathcal{H}_X)\rightarrow (\mathcal{H}_{S}\otimes\mathcal{H}_X)$ to be
\begin{equation}\label{eq:U_sx}
    U_{SX}=U_\theta\otimes\operatorname{id}_X,
\end{equation}
where $U_\theta$ {is as in} Eq.~(\ref{Eq:PhaseO}). Moreover, consider an arbitrary joint density operator $\sigma_{SX}\in\mathcal{D}(\mathcal{H}_S\otimes\mathcal{H}_X)$ {and denote}
\begin{equation}\label{eq:sigma^theta}
    \sigma_{SX}^\theta{:=} U_{SX}\sigma_{SX}U^\dagger_{SX}.
\end{equation}
Due to the invariance property of von-Neumann entropy under unitary transformations, we have

\begin{align}\label{eq:phase_invariant_sq}
     \mathcal{S}({S}|{E'})_{{\sigma_{SE'}}}+\mathcal{S}({S}|{F'})_{{\sigma_{SF'}}}=
\mathcal{S}({S}|{E'})_{{\sigma^\theta_{SE'}}}+\mathcal{S}({S}|{F'})_{{\sigma^\theta_{SF'}}},
\end{align}

where $\sigma_{SE'}^\theta$ and 
$\sigma_{SF'}^\theta$ are defined in Eq.~(\ref{eq:sigma^theta}) with the Hilbert space $\mathcal{H}_X$ to be $\mathcal{H}_{E'}$ and $\mathcal{H}_{F'}$, respectively, and
\begin{align}
    &\sigma_{SE'}{:=}(\operatorname{id}_{S}\otimes\mathcal{N}^{sq}_{{E\rightarrow E'}})(U_{\mathcal{N}^{\gamma}_{{S\rightarrow S}}}\rho_{SE}U_{\mathcal{N}^{\gamma}_{{S\rightarrow S}}}^\dagger),\\
    &\sigma_{SF'}{:=}(\operatorname{id}_{S}\otimes\mathcal{N}_{{E\rightarrow E'}}^{sq^{ {c} }})(U_{\mathcal{N}^{\gamma}_{{S\rightarrow S}}}\rho_{SE}U_{\mathcal{N}_{{S\rightarrow S}}^{\gamma}}^\dagger),
\end{align}
with $\rho_{SE}{\in\mathcal{T}(\mathcal{H}_S\otimes\mathcal{H}_E)}$ being an arbitrary {system-environment} initial state. On the other hand, the conditional entropy is {concave, meaning that the following relation holds true:}
\begin{align}\label{eq:concav_cond_ent}
    \frac{1}{2\pi}\int_{0}^{2\pi} d\theta\Big(&\mathcal{S}({S}|{E'})_{{\sigma^\theta_{SE'}}}+\mathcal{S}({S}|{F'})_{{\sigma^\theta_{SF'}}}\Big)\notag\\
    &\leq \mathcal{S}({S}|{E'})_{\bar{\sigma}_{SE'}}+\mathcal{S}({S}|{F'})_{\bar{\sigma}_{SF'}},
\end{align}
where 
\begin{align}\label{eq:sigma_SE'}
    &\bar{\sigma}_{SE'}{:=} (\operatorname{id}_S\otimes\mathcal{N}^{sq}_{E\rightarrow  E'})(U_{\mathcal{N}^\gamma_{S\rightarrow  S}}\Big(\frac{1}{2\pi}\int_0^{2\pi}d\theta\rho_{SE}^\theta\Big) U_{\mathcal{N}^\gamma_{S\rightarrow  S}}^\dagger),
\end{align}   
\begin{align}\label{eq:sigma_SF'}
    &\bar{\sigma}_{SF'}{:=} (\operatorname{id}_S\otimes\mathcal{N}^{sq^{{c}}}_{E\rightarrow  F'})(U_{\mathcal{N^\gamma_{S\rightarrow  S}}}\Big(\frac{1}{2\pi}\int_0^{2\pi}d\theta\rho_{SE}^\theta\Big) U_{\mathcal{N}^{\gamma}_{S\rightarrow  S}}^\dagger),
\end{align}
with $\rho_{SE}^\theta$ defined in the same way as in Eq.~(\ref{eq:sigma^theta}). Considering  Eq.~(\ref{eq:unitary_d}) and Eq.~(\ref{eq:U_sx}), with the aid of simple algebraic steps, it can be seen that the following commutation relation holds true:
\begin{equation}
\label{eq:CommutU}
    [U_{\mathcal{N}^\gamma_{S\rightarrow  S}},U_{SE}]=0.
\end{equation}
It means that the unitary extension of the phase covariant bosonic dephasing channel is invariant under local phase operator and is symmetric.
Using this commutation relation we can conclude that:
\begin{align}
&U_{SE'}(\operatorname{id}_{S}\otimes\mathcal{N}^{sq}_{E\rightarrow  E'})U_{\mathcal{N}^\gamma_{S\rightarrow  S}}=(\operatorname{id}_{S}\otimes\mathcal{N}^{sq}_{E\rightarrow  E'})U_{\mathcal{N}^\gamma_{S\rightarrow  S}}U_{SE},\\
&U_{SF'}(\operatorname{id}_{S}\otimes\mathcal{N}^{sq^{{c}}}_{E\rightarrow  F'})U_{\mathcal{N}^\gamma_{S\rightarrow  S} }=(\operatorname{id}_{S}\otimes\mathcal{N}^{sq^{{c}}}_{E\rightarrow  F'})U_{\mathcal{N}^\gamma_{S\rightarrow  S}}U_{SE}.
\end{align}
{Then, by means of Eq.~(\ref{eq:phase_invariant_sq}), the relation} \eqref{eq:concav_cond_ent} becomes:
\begin{equation}\label{eq:final_prop}
    \mathcal{S}(S|E')_{\sigma_{SE'}}+\mathcal{S}(S|F')_{\sigma_{SF'}}
    \leq \mathcal{S}(S|E')_{\bar{\sigma}_{SE'}}+\mathcal{S}(S|F')_{\bar{\sigma}_{SF'}}.
\end{equation}
{Now,} consider the initial joint state of the system and the environment to be
\begin{equation}\label{eq:genetic_input}
    \rho_{SE}=\rho_S\otimes\ket{0}_E\bra{0}.
\end{equation}
By expanding {$\rho_S$ in the Fock basis as $\rho_S=\sum_{n,m}\rho_{m,n} \ket{n}\bra{m}$,} the integrals in Eq.~(\ref{eq:sigma_SE'}) and Eq.~(\ref{eq:sigma_SF'}) take the following form:
\begin{align}\label{eq:optimal_input}
    &\frac{1}{2\pi}\int_0^{2\pi}d\theta\rho_{SE}^\theta=\frac{1}{2\pi}\int_0^{2\pi}d\theta U_{SE}(\rho_S\otimes\ket{0}_E\bra{0})U^\dagger_{SE}\notag\\
    &=\frac{1}{2\pi}\sum_{m,n=0}^\infty\int_{0}^{2\pi}d\theta \rho_{n,m}e^{\operatorname{i}(n-m)\theta}\ket{n}_S\bra{m}\otimes\ket{0}_E\bra{0}\notag\\
    &=\sum_{n=0}^\infty\rho_{n,n}\ket{n}_S\bra{n}\otimes\ket{0}_E\bra{0}.
\end{align}
Therefore, by considering the above relation along with Eq.~(\ref{eq:final_prop}), it can be seen that the optimal input state for {squashed entanglement of the channel} defined in Eq.~(\ref{eq:channel_squashing_alt}) is {diagonal in the Fock basis.}

For the case where we have energy constraint on the input states the same {arguments} from Eq.~(\ref{eq:U_sx}) to Eq.~(\ref{eq:optimal_input}) holds {true. However,} this time the optimal input state in Eq.(\ref{eq:optimal_input}) takes the form
\begin{equation}\label{eq:energy_cons_optimal_input}
    \rho_{SE}^{opt}{={\sum_{n}}'}\rho_{n,n}\ket{n}_S\bra{n}\otimes\ket{0}_E\bra{0},
\end{equation}
where 
{
\begin{equation}
{\sum_{n}}':=\sum\limits_{\substack{n=0 \\ {\rm Tr}(a^\dagger a \rho_S)\leq N}}^\infty,
\end{equation}
and we shall use this notation hereafter.}
\end{proof}

The proof is based on two main properties. The first one is the concavity of conditional entropy. Similar arguments have been used to bound the squashed entanglement of other channels \cite{wildelocc}.
The second one, is the symmetric property of the unitary extension of the phase covariant 
bosonic dephasing channel without invoking the fact that the isometric extension of a group covariant channel, has covariant property \cite{DasBaumlWilde2020}.

{Thanks to Proposition} \ref{prop:fock_optimal}, the supremum in Eq.~(\ref{eq:ecsq}) is replaced by a supremum over the set of diagonal states in {the Fock} basis satisfying the energy constraint, or in other words, over {the probability distributions of Fock} states satisfying the energy constraint:
\begin{equation}
    \label{eq:OptInt}
    \tilde{E}_{sq}(\mathcal{N}^\gamma_{S\rightarrow S},a^\dagger a, N)= \sup_{p_n:\sum_n n p_n\leq N}E_{sq}\Big(\rho_S^{opt},\mathcal{N}^\gamma_{S\rightarrow S}\Big),
\end{equation}
where 
\begin{equation}
\label{eq:rhoOpt}
    \rho_S^{opt}={{\sum_n}'} p_n\ket{n}_S\bra{n},\quad p_n:=\rho_{n,n},
\end{equation}
is obtained by tracing over the environment degrees of freedom of Eq.~(\ref{eq:energy_cons_optimal_input}).
{Hence, for the optimal input state, the system-environment output state is given by
\begin{equation}
\label{eq:SigmaSEopt}
\sigma_{SE'}^{opt}={{\sum_{n}}'} p_n \ket{n}_{S}\bra{n}\otimes
    \ket{-\operatorname{i}\sqrt{\gamma}n}_E\bra{-\operatorname{i}\sqrt{\gamma}n}.
\end{equation} 
}
{For subsequent developments,} it is more convenient to re-express
Eq.~(\ref{eq:not_for_inf}) in terms of mutual information, namely:
\begin{align}\label{eq:sq_mutual_inf}
 &E_{sq}(\rho_S,\mathcal{N}_{S\rightarrow S'})\cr
 &=\inf_{\mathcal{V}_{E\rightarrow E'F'}^{\mathcal{N}^{sq}_{E\rightarrow E'}}}\frac{1}{2}\big( \mathcal{S}(S'|E')_{\sigma_{S{'}E'}}+\mathcal{S}(S{'}|F')_{\sigma_{S{'}F'}}\big)\notag\\
 &=\inf_{\mathcal{V}_{E\rightarrow E'F'}^{\mathcal{N}^{sq}_{E\rightarrow E'}}}\frac{1}{2}\Big(\mathcal{S}(\sigma_{S{'}E'})-\mathcal{S}(\sigma_{E'})\notag+\mathcal{S}(\sigma_{S{'}F'})-\mathcal{S}(\sigma_{F'})\Big)\notag\\
 &=\mathcal{S}(\sigma_{S{'}})-\sup_{\mathcal{V}_{E\rightarrow E'F'}^{\mathcal{N}^{sq}_{E\rightarrow E'}}}\frac{1}{2}\Big(I(S{'};E')_{\sigma_{S{'}E'}}+I(S{'};F')_{\sigma_{S{'}F'}}\Big),
\end{align}
Therefore, for the bosonic dephasing channel with optimal input state we have:
\begin{align}\label{eq:squashed_entangl}
& E_{sq}(\rho_S^{opt},\mathcal{N}^\gamma_{S\rightarrow S})=\cr
 &\mathcal{S}(\sigma_{S}^{opt})-\sup_{\mathcal{V}_{E\rightarrow E'F'}^{\mathcal{N}^{sq}_{E\rightarrow E'}}}\frac{1}{2}\Big(I(S;E')_{\sigma^{{opt}}_{SE'}}+I(S;F')_{\sigma^{{opt}}_{SF'}}\Big),\cr
\end{align}
where due to the invariance of optimal input state under channel action, the output of the channel is given by $\sigma_S^{opt}=\rho_S^{opt}$ and
\begin{align}
    &\sigma_{SE'}^{opt}={{\sum_{n}}'} p_n \ket{n}_{S}\bra{n}\otimes{\mathcal{N}^{sq}_{E\rightarrow E'}}
    \Big(\ket{-\operatorname{i}\sqrt{\gamma}n}\bra{-\operatorname{i}\sqrt{\gamma}n}\Big),
\label{eq:dephasing_output}\\
    &\sigma_{SF'}^{opt}={{\sum_{n}}'} p_n \ket{n}_{S}\bra{n}\otimes\mathcal{N}^{sq^{{c}}}_{E\rightarrow F'}\Big(\ket{-\operatorname{i}\sqrt{\gamma}n}\bra{-\operatorname{i}\sqrt{\gamma}n}\Big),
    \label{eq:dephasing_outputSF'}
\end{align}
are classical-quantum states.

{In this section we investigated the squashed entanglement of a bosonic dephasing channel which is an upper bound for its energy-constrained two-way LOCC assisted quantum capacity (Eq.~(\ref{eq:TLOCCAQEC})). We showed that to compute this upper bound, two optimizations are required: one over the probability distribution of
Fock states at the input (see Eq.~(\ref{eq:OptInt})) and the other over isometric extensions of squashing channel (see Eq.~(\ref{eq:squashed_entangl})). In the next section, we discuss the optimization over isometric extensions of squashing channels in
Eq.~(\ref{eq:not_for_inf}).

\section{Squashing channels}
\label{sec:squashing}

In this section, first, we recall the definition of quantum Markov chain (QMC) for tri-partite quantum states \cite{winterc,Petz1,Petz2}. Then we focus on a subset of QMC states which we call symmetric quantum Markov chains (SQMC) due to a particular symmetry they have. Subsequently we define quantum Markov chain inducer (QMCI) channels, and symmetric quantum Markov chain inducer (SQMCI) channels, which are quantum channels that transform an input state to respectively QMC state and SQMC state. Then we show that
if for system-environment output of the channel, the set of SQMCI channels is not null, then 
the optimal squashing channel belongs to this set of SQMCI channels.

{\begin{definition}\label{def:markovchain} (\cite{winterc,Petz1,Petz2})
A tri-partite state $\sigma_{SE'F'}$ is a quantum Markov chain (QMC) with the order $S\leftrightarrow E'\leftrightarrow F'$ if and only if there exists a recovery channel $\mathcal{R}_{E'\rightarrow E'F'}$, such that
\begin{equation}
    \sigma_{SE'F'}=(\operatorname{id}_S\otimes\mathcal{R}_{E'\rightarrow E'F'})(\sigma_{SE'}),
\end{equation}
where $\sigma_{SE'}=\operatorname{Tr}_{F'}(\sigma_{SE'F'})$.
\end{definition}
{The following lemma introduces a convenient way to verify if a given tri-partite state is a QMC state:}
\begin{lemma}\label{lemma:condMu}
A tri-partite state $\sigma_{SE'F'}$ is a QMC with the order $S\leftrightarrow E'\leftrightarrow F'$ if and only if the following relation holds \cite{winterc}
\begin{equation}\label{eq:markovind}
    I(S;F'|E')_{\sigma_{SE'F'}}=0.
\end{equation}
\end{lemma}
Lemma~\ref{lemma:condMu} implies that, if a tri-partite state $\sigma_{SE'F'}$ is a QMC with the order $S\leftrightarrow E'\leftrightarrow F'$ then the marginal state $\sigma_{SF'}$ is a separable state \cite{winterc,Petz1,Petz2}. Conversely, if $\sigma_{SF'}$ is separable, there exists an extension of it
$\sigma_{SE'F'}$ ($\sigma_{SF'}=\operatorname{Tr}_{E'}(\sigma_{SE'F'})$) such that $I(S;F'|E')_{\sigma_{SE'F'}}=0$. For an example of QMC state with order of
$S\leftrightarrow E'\leftrightarrow F'$ consider the following tri-partite state
\begin{equation}
\label{eq:ExQMC}
 \sigma_{SE'F'}=\sum_{n}p_n \ket{s_n}\bra{s_n}\otimes \ket{e_n}\bra{e_n}\otimes\eta_n,
\end{equation}
where $p_n>0$ for all $n$, $\sum_{n}p_n=1$, $\bra{s_n}s_{n'}\rangle=\delta_{n,n'}$, $\bra{e_n}e_{n'}\rangle=\delta_{n,n'}$ and $\eta_n\in\mathcal{T}(\mathcal{H}_{F'})$ are density matrices. It is easy to see that for the state in Eq.~(\ref{eq:ExQMC}) we have $I(S;F'|E')=0$. Similarly, one can show that the state
\begin{equation}
\label{eq:ExQMC2}
 \sigma_{SE'F'}=\sum_{n}p_n \ket{s_n}\bra{s_n}\otimes\xi_n\otimes \ket{f_n}\bra{f_n}
\end{equation}
with $\bra{f_n}f_{n'}\rangle=\delta_{n,n'}$ and $\xi_n\in\mathcal{T}(\mathcal{H}_{E'})$ is a QMC with the order $S\leftrightarrow F'\leftrightarrow E'$

{In the next definition, we impose more constraints on a QMC  tri-partite state by demanding symmetric properties in the order the state is a QMC:}
\begin{definition}\label{def:SQMC}
A tri-partite state $\sigma_{SE'F'}$ is a Symmetric quantum Markov chain (SQMC), if it is a QMC with the order $S\leftrightarrow E'\leftrightarrow F'$, and $S\leftrightarrow F'\leftrightarrow E'$.
\end{definition}
According to Lemma.~{\ref{lemma:condMu}}, if the tri-partite state $\sigma_{SE'F'}$ is a SQMC state 
with the order $S\leftrightarrow E'\leftrightarrow F'$, and $S\leftrightarrow F'\leftrightarrow E'$ then the marginal states
$\sigma_{SE'}$ and $\sigma_{SF'}$ are separable states.
Also according to Definition ~(\ref{def:markovchain}), it is easy to see that given a SQMC tri-partite state, there exist recovery channels $\mathcal{R}_{E'\rightarrow E'F'}$ and $\mathcal{R'}_{F'\rightarrow E'F'}$, such that
\begin{align}
    &\sigma_{SE'F'}=(\operatorname{id}_S\otimes\mathcal{R}_{E'\rightarrow E'F'})(\sigma_{SE'})\cr
    &\sigma_{SE'F'}=(\operatorname{id}_S\otimes\mathcal{R'}_{F'\rightarrow E'F'})(\sigma_{SF'}),
\end{align}
where $\sigma_{SE'}=\operatorname{Tr}_{F'}(\sigma_{SE'F'})$, and  $\sigma_{SF'}=\operatorname{Tr}_{E'}(\sigma_{SE'F'})$.
The 
quantum state 
\begin{equation}
\label{eq:exampleSQMCI}
    \sigma_{SE'F'}=\sum_{n}p_n\Omega_n\otimes \ket{e_n}\bra{e_{n}}\otimes\ket{f_n}\bra{f_{n}}
\end{equation}
with $\Omega_n\in\mathcal{T}(\mathcal{H}_S)$, $p_n>0, \forall n$, $\sum_{n}p_n=1$ and 
\begin{equation}
   \bra{e_n}e_{n'}\rangle=\bra{f_n}f_{n'}\rangle=\delta_{n,n'}
\end{equation}
is an example of SQMC state.

{In the following, we discuss channels that can convert input states to QMC or SQMC states.}
\begin{definition}
A quantum channel $\mathcal{N}_{E\rightarrow E'}$ with  corresponding isometry conjugation $\mathcal{V}^{\mathcal{N}_{E\rightarrow E'}}_{E\rightarrow E'F'}$ is a quantum Markov chain inducer (QMCI) channel for a bipartite initial state $\sigma_{SE}$, if the tri-partite state
\begin{equation}
    \label{eq:inducer}
    \sigma_{SE'F'}=(\operatorname{id}_S\otimes \mathcal{V}^{\mathcal{N}_{E\rightarrow E'}}_{E\rightarrow E'F'})(\sigma_{SE}).
\end{equation}
is a QMC with the order $S\leftrightarrow E'\leftrightarrow.F'$
\end{definition}

An example of QMCI channel for a bipartite state $\sigma_{SE}$ is the one that transforms it to a product state $\sigma_{SE}\otimes\sigma_{F'}$, a QMC state with order of $S\leftrightarrow E\leftrightarrow F'$. Such a channel is described by an isometry over the extended state:
\begin{equation}
    V_{E\rightarrow E'F'}=\sum_{n,i}\sqrt{\lambda_i} \ket{e_n,\lambda_i}\bra{\lambda_i}
\end{equation}
where $\{\ket{e_n}\}$ forms an orthogonal set in $\mathcal{H}_{E}$ and in its isometric space $\mathcal{H}_{E'}$, while $\{\lambda_i, \ket{\lambda_i}\}$ are eigenvalues and eigenvectors of the quantum state $\sigma_{F'}\in\mathcal{T}(\mathcal{H}_{F'})$. 
 
The definition of QMCI channels can be  generalized to channels that produce QMC with symmetric order:

\begin{definition}\label{def:sqmci}
A quantum channel $\mathcal{N}_{E\rightarrow E'}$ with  corresponding isometry conjugation $\mathcal{V}^{\mathcal{N}_{E\rightarrow E'}}_{E\rightarrow E'F'}$ is a symmetric quantum Markov chain inducer (SQMCI) channel for a bipartite initial state $\sigma_{SE}$, if 
the tri-partite state
\begin{equation}
    \sigma_{SE'F'}=(\operatorname{id}_S\otimes \mathcal{V}^{\mathcal{N}_{E\rightarrow E'}}_{E\rightarrow E'F'})(\sigma_{SE})
\end{equation}
is a SQMC with the order $S\leftrightarrow E'\leftrightarrow F'$ and $S\leftrightarrow F'\leftrightarrow E'$.
\end{definition}

For an example of SQMCI channel for separable initial state
\begin{equation}
\label{eq:iniSep}
    \sigma_{SE}=\sum_n p_n \mu_n\otimes \ket{e_n}\bra{e_n}
\end{equation}
with $\mu_n\in\mathcal{T}(\mathcal{H}_S)$, positive $p_n, \forall n$, $\sum_n p_n=1$ and with orthonormal basis $\{\ket{e_n}\}$, consider an isometry
\begin{equation}
\label{eq:ExQMCI}
    V_{E\rightarrow
    EF'}=\sum_l\ket{e_l,
    f_l}\bra{e_l}.
\end{equation}
The separable state $\sigma_{SE}$ in Eq.~(\ref{eq:iniSep}) under the isometry (\ref{eq:ExQMCI}) is transformed to
\begin{equation}
    \sigma_{SEF'}=\sum_n p_n \mu_n\otimes \ket{e_n}\bra{e_n}\otimes \ket{f_n}\bra{f_n},
\end{equation}
which is a SQMC state with both orders $S\leftrightarrow E\leftrightarrow F'$ and $S\leftrightarrow F'\leftrightarrow E$.
 \begin{proposition} 
 \label{prop:2}
 Consider a channel $\mathcal{N}_{S\rightarrow S'}$ with an isometry conjugation $\mathcal{V}^{\mathcal{N}_{S\rightarrow S'}}_{S\rightarrow S'E}$, and a generic bipartite system-environment output density operator $\sigma_{S'E}$. If the set of SQMCI channels acting on the subspace $\mathcal{H}_{E}$ is not null, it contains the optimal squashing channel.
 \end{proposition}

\begin{proof}

Consider a generic squashing channel $\mathcal{N}^{sq}_{E\rightarrow E'}$,  and $\widetilde{\mathcal{N}}^{sq}_{E\rightarrow E'}$ a squashing channel which is SQMCI with initial state $\sigma_{S'E}$ and denote their isometry conjugations respectively by $\mathcal{V}^{{{\cal N}}^{sq}_{E\rightarrow E'}}_{E\rightarrow E'F'}$ and $\mathcal{V}^{{\widetilde{\cal N}}^{sq}_{E\rightarrow E'}}_{E\rightarrow E'F'}$. There always exists unitary conjugation $\mathcal{U}_{E'F'\rightarrow E'F'}$ such that
\begin{equation}
    \mathcal{U}_{E'F'\rightarrow E'F'}\circ\mathcal{V}^{{{\cal N}}^{sq}_{E\rightarrow E'}}_{E\rightarrow E'F'}=\mathcal{V}^{{\widetilde{\cal N}}^{sq}_{E\rightarrow E'}}_{E\rightarrow E'F'}
\end{equation} 
As $\mathcal{U}_{E'F'\rightarrow E'F'}$ acts only on $\mathcal{H}_E\otimes\mathcal{H}_F$  the correlation between $\sigma_{S'}\in\mathcal{T}(\mathcal{H}_{S'})$ and $\sigma_{E'F'}\in\mathcal{T}(\mathcal{H}_{E'}\otimes\mathcal{H}_{F'})$ remains intact under the action of $\mathcal{U}_{E'F'\rightarrow E'F'}$. Therefore, 
\begin{equation}
\label{equality0}
I(S;E'F')_{\mathcal{V}^{{\widetilde{\cal N}}^{sq}_{E\rightarrow E'}}_{E\rightarrow E'F'}(\sigma_{SE})}=I(S;E'F')_{\mathcal{V}^{{{\cal N}}^{sq}_{E\rightarrow E'}}_{E\rightarrow E'F'}(\sigma_{SE})}.
\end{equation}
Using the definition of conditional mutual information in (\ref{eq:CondMut}), Eq.~(\ref{equality0}) can also be written as
\begin{align}
\label{equality1}
&I(S;E')_{\mathcal{V}^{{\widetilde{{\cal N}}}_{E\rightarrow E'}}_{E\rightarrow E'F'}(\sigma_{SE})}+I(S;F'|E')_{\mathcal{V}^{{\widetilde{{\cal N}}}_{E\rightarrow E'}}_{E\rightarrow E'F'}(\sigma_{SE})}
=\cr
&I(S;E')_{\mathcal{V}^{{{\cal N}}^{sq}_{E\rightarrow E'}}_{E\rightarrow E'F'}(\sigma_{SE})}+I(S;F'|E')_{\mathcal{V}^{{{\cal N}}^{sq}_{E\rightarrow E'}}_{E\rightarrow E'F'}(\sigma_{SE})}.
\end{align}
or as
\begin{align}
\label{eq:SQMC2}
&I(S;F')_{\mathcal{V}^{{\widetilde{{\cal N}}}_{E\rightarrow E'}}_{E\rightarrow E'F'}(\sigma_{SE})}+I(S;E'|F')_{\mathcal{V}^{{\widetilde{{\cal N}}}_{E\rightarrow E'}}_{E\rightarrow E'F'}(\sigma_{SE})}
=\cr
&I(S;F')_{\mathcal{V}^{{{\cal N}}^{sq}_{E\rightarrow E'}}_{E\rightarrow E'F'}(\sigma_{SE})}+I(S;E'|F')_{\mathcal{V}^{{{\cal N}}^{sq}_{E\rightarrow E'}}_{E\rightarrow E'F'}(\sigma_{SE})}.
\end{align}
As the channel $\widetilde{\mathcal{N}}_{E\rightarrow E'}$, is a SQMCI, by using Lemma.~(\ref{lemma:condMu}), Eqs.~(\ref{equality1}) and (\ref{eq:SQMC2}), result
\begin{align}
\label{eq:SQMC1-2}
&I(S;E')_{\mathcal{V}^{{\widetilde{{\cal N}}}_{E\rightarrow E'}}_{E\rightarrow E'F'}(\sigma_{SE})}
=\cr
&I(S;E')_{\mathcal{V}^{{{\cal N}}^{sq}_{E\rightarrow E'}}_{E\rightarrow E'F'}(\sigma_{SE})}+I(S;F'|E')_{\mathcal{V}^{{{\cal N}}^{sq}_{E\rightarrow E'}}_{E\rightarrow E'F'}(\sigma_{SE})},
\end{align}
and
\begin{align}
\label{eq:SQMC2-2}
&I(S;F')_{\mathcal{V}^{{\widetilde{{\cal N}}}_{E\rightarrow E'}}_{E\rightarrow E'F'}(\sigma_{SE})}
=\cr
&I(S;F')_{\mathcal{V}^{{{\cal N}}^{sq}_{E\rightarrow E'}}_{E\rightarrow E'F'}(\sigma_{SE})}+I(S;E'|F')_{\mathcal{V}^{{{\cal N}}^{sq}_{E\rightarrow E'}}_{E\rightarrow E'F'}(\sigma_{SE})}.
\end{align}
Taking into account that the conditional mutual information is a non-negative quantity Eqs.~(\ref{eq:SQMC1-2}) and (\ref{eq:SQMC2-2}) give the following inequalities:
\begin{align}
&I(S;E')_{\mathcal{V}^{{\widetilde{{\cal N}}}_{E\rightarrow E'}}_{E\rightarrow E'F'}(\sigma_{SE})}
\geq
I(S;E')_{\mathcal{V}^{{{\cal N}}^{sq}_{E\rightarrow E'}}_{E\rightarrow E'F'}(\sigma_{SE})}\cr
&I(S;F')_{\mathcal{V}^{{\widetilde{{\cal N}}}_{E\rightarrow E'}}_{E\rightarrow E'F'}(\sigma_{SE})}
\geq
I(S;F')_{\mathcal{V}^{{{\cal N}}^{sq}_{E\rightarrow E'}}_{E\rightarrow E'F'}(\sigma_{SE})}
\end{align}
Therefore, 
\begin{align}
\label{eq:SQMCI>}
&I(S;E')_{\mathcal{V}^{{\widetilde{{\cal N}}}_{E\rightarrow E'}}_{E\rightarrow E'F'}(\sigma_{SE})}+I(S;F')_{\mathcal{V}^{{\widetilde{{\cal N}}}_{E\rightarrow E'}}_{E\rightarrow E'F'}(\sigma_{SE})}
\geq\cr
&I(S;E')_{\mathcal{V}^{{{\cal N}}^{sq}_{E\rightarrow E'}}_{E\rightarrow E'F'}(\sigma_{SE})}
+
I(S;F')_{\mathcal{V}^{{{\cal N}}^{sq}_{E\rightarrow E'}}_{E\rightarrow E'F'}(\sigma_{SE})}
\end{align}
Regarding the supremum in Eq.~(\ref{eq:sq_mutual_inf}) from Eq.~(\ref{eq:SQMCI>}) we conclude that the optimal squashing channel belongs to the set of SQMCI channels for the state 
$\sigma_{S'E}$ if such a set is not null.
\end{proof}
Indeed Proposition~\ref{prop:2} is useful to restrict the set over which optimization for finding optimal squashing channel is done, conditioned to the fact that for the channel output over extended Hilbert space $\mathcal{H}_{S'}\otimes\mathcal{H}_{E}$ there exists SQMCI channels. Namely, when the set of SQMCI squashing channels is not empty.


So far, we constructed the set over which the optimization for squashing channel must be taken, the following proposition simplify the quantity defined in the last line of Eq.~(\ref{eq:sq_mutual_inf}) which we are going to optimize.
\begin{proposition}\label{prop:mut_equal_qmci}
For any SQMCI channel $\mathcal{N}_{E\rightarrow E'}$ with initial state $\rho_{SE}$ the following equality holds
\begin{align}
\frac{1}{2}\Big(I(S;E')_{\sigma_{SE'}}+I(S;F')_{\sigma_{SF'}}\Big)=I(S;E')_{\sigma_{SE'}},
\end{align}
where the entropic quantities are calculated over the state $\sigma_{SE'F'}$ defined as:
\begin{equation}
    \sigma_{SE'F'}=(\operatorname{id}_S\otimes\mathcal{V}^{\mathcal{N}_{E\rightarrow E'}}_{E\rightarrow E'F'})(\rho_{SE}),
\end{equation}
with $\mathcal{V}^{\mathcal{N}_{E\rightarrow E'}}_{E\rightarrow E'F'}$ being an isometric  extension of the channel $\mathcal{N}_{E\rightarrow E'}$.
\begin{proof}
By assumption, the channel $\mathcal{N}_{E\rightarrow E'}$ with isometry conjugation $\mathcal{V}^{\mathcal{N}_{E\rightarrow E'}}_{E\rightarrow E'F'}$ is a SQMCI channel with initial state $\rho_{SE}$. Hence
\begin{equation}
    \sigma_{SE'F'}=\mathcal{V}^{\mathcal{N}_{E\rightarrow E'}}_{E\rightarrow E'F'}[\rho_{SE}]
\end{equation}
is a SQMC state. Therefore, there exist channels 
$\mathcal{R}_{E'\rightarrow E'F'}$, and $\mathcal{R}'_{F'\rightarrow E'F'}$, such that 
\begin{align}
    &\sigma_{SE'F'}=(\operatorname{id}_S\otimes\mathcal{R}_{E'\rightarrow E'F'})(\sigma_{SE'})\cr
    &\sigma_{SE'F'}=(\operatorname{id}_S\otimes\mathcal{R'}_{F'\rightarrow E'F'})(\sigma_{SF'}),
\end{align}
where 
$\sigma_{SE'}=\operatorname{Tr}_{F'}(\sigma_{SE'F'})$, and  $\sigma_{SF'}=\operatorname{Tr}_{E'}(\sigma_{SE'F'})$.

By defining 
\begin{align}
   &\mathcal{M}_{E'\rightarrow F'}\coloneqq \operatorname{Tr}_{E'}\circ\mathcal{R}_{E'\rightarrow E'F'},\cr &\mathcal{M}'_{F'\rightarrow E'}\coloneqq \operatorname{Tr}_{F'}\circ\mathcal{R}'_{F'\rightarrow E'F'},
\end{align}

and according to the data processing argument, the following inequalities hold:
\begin{align}\label{ineq:n_nc}
    &I(S;E')_{(\mathcal{{M}}'_{F'\rightarrow E'}\circ\mathcal{N}^c_{E\rightarrow F'})(\rho_{SE})}\le I(S;F')_{\mathcal{N}^c_{E\rightarrow F'}(\rho_{SE})}\cr
    &I(S;F')_{(\mathcal{{M}}_{E'\rightarrow F'}\circ\mathcal{N}_{E\rightarrow E'})(\rho_{SE})}\le I(S;E')_{\mathcal{N}_{E\rightarrow E'}(\rho_{SE})}.\cr
\end{align}
On the other hand, the action of the SQMCI channel $\mathcal{N}_{E\rightarrow E'}$ and its complement $\mathcal{N}^c_{E\rightarrow F'}$ on $\rho_{SE}$ is given by
\begin{align}\label{eq:sym}
    &\mathcal{N}_{E\rightarrow E'}(\rho_{SE})=\mathcal{(}\mathcal{M}'_{F'\rightarrow E'}\circ\mathcal{N}^c_{E\rightarrow F'}\big)(\rho_{SE})\cr
    &\mathcal{N}^c_{E\rightarrow F'}(\rho_{SE})=\big(\mathcal{M}_{E'\rightarrow F'}\circ\mathcal{N}_{E\rightarrow E'}\big)(\rho_{SE}),
\end{align}
Combining Eq.~(\ref{eq:sym}) with Eq.(\ref{ineq:n_nc}), we have 
\begin{equation}
    I(S;F')_{\mathcal{N}^c_{E\rightarrow F'}(\rho_{SE})}=I(S;E')_{\mathcal{N}_{E\rightarrow E'}(\rho_{SE})},
\end{equation}
Therefore,
\begin{align}
    \frac{1}{2}\Big(I(S;F')_{\mathcal{N}^c_{E\rightarrow F'}(\rho_{SE})}&+I(S;E')_{\mathcal{N}_{E\rightarrow E'}(\rho_{SE})}\Big)\cr
    &=I(S;E')_{\mathcal{N}_{E\rightarrow E'}(\rho_{SE})}.
\end{align}
\end{proof}
\end{proposition}
This proposition simplifies computing Eq.(\ref{eq:sq_mutual_inf}) when the optimal squashing channel is a SQMCI channel. 

\section{Squashing channel for bosonic dephasing channel}
\label{sec:examples}
{In this section we analyze two examples of squashing channels for bosonic dephasing channel. First, we explain why we chose these  examples from the set of symmetric channels. Then in the following subsections we study in details 50/50 beamsplitter, and symmetric qubit channels as squashing channel. 

For quantum dephasing channel, the output of optimal input state, $\sigma^{opt}_{SE}$ is a separable state as given in
Eq.~(\ref{eq:SigmaSEopt}).
It is worth noticing that if a SQMCI channel, $\mathcal{N}_{E\rightarrow E'}$ with isometry extension $\mathcal{V}^{\mathcal{N}_{E\rightarrow E'}}_{E\rightarrow E'F'}$ exists, it transforms $\sigma^{opt}_{SE}$ to $\sigma^{opt}_{SE'F'}$ for which 
\begin{equation}
\label{eq:StrongSubAd}
    I(S;E'|F')_{\sigma^{opt}_{SE'F'}}=I(S;F'|E')_{\sigma^{opt}_{SE'F'}}=0
\end{equation}
Satisfying Eq.~(\ref{eq:StrongSubAd}) by $\sigma^{opt}_{SE'F'}$ is equivalent to satisfying strong subadditivity property with equality by $\sigma_{SE'F'}$ \cite{winterc}. A natural structure for such states is given by
\begin{equation}
    \sigma^{opt}_{SE'F'}=\sum_n p_n\ket{n}\bra{n}\otimes \ket{e_n}\bra{e_n}\otimes\ket{f_n}\bra{f_n}
\end{equation}
where $\bra{e_m} e_n\rangle=\bra{f_m} f_n\rangle=\delta_{m,n}$. These orthogonality properties play crucial role in satisfying  Eq.~(\ref{eq:StrongSubAd}). 
But as inner product is invariant under isometry and taking into account that coherent states are non-orthogonal, it is impossible to find an isometry such that $V \ket{-\operatorname{i}\sqrt{\gamma}n}_E=\ket{e_n,f_n}_{E'F'}$. Hence, we conjecture that the set of SQMCI channel for density operator $\sigma^{opt}_{SE}$ in Eq.~(\ref{eq:SigmaSEopt}) is empty.
Therefore we switch to squashing channels having a structure close to SQMCI channels, namely the symmetric channels and analyze  Eq.~(\ref{eq:squashed_entangl}) where $\mathcal{N}_{E\rightarrow E'}$ is a symmetric channel. 

Confining our search for squashing channels to the set of symmetric channels, Eq.~(\ref{eq:squashed_entangl}) turns into:
\begin{align}
    &E_{sq} (\rho_S^{opt},{\mathcal{N}}^\gamma_{S\rightarrow S})\notag\\
    &\leq\mathcal{S}(\sigma_S^{opt})-\sup_{\mathcal{V}^{\mathcal{N}^{sq}_{E\rightarrow E'}}_{E\rightarrow E'F'}\in{\mathbb{V}}_{sym}}\frac{1}{2}\Big(I(S;E')_{\sigma^{opt}_{SE'}}+I(S;F')_{\sigma^{opt}_{SF'}}\Big)
\end{align}
where $\sigma_{SE'}^{opt}$ is defined in Eq.~(\ref{eq:dephasing_output}), and $\mathbb{V}_{sym}$ is the set of isometric dilation of symmetric channels. Obviously, when the squashing channel belongs to the set of symmetric channels $\sigma_{SE'}^{opt}=\sigma_{SF'}^{opt}$. Hence we have
\begin{align}
\label{eq:self_comp_sq}
    E_{sq} (\rho_S^{opt},{\mathcal{N}}^\gamma_{S\rightarrow S})
    &\leq\mathcal{S}(\sigma_S^{opt})-\sup_{\mathcal{V}^{\mathcal{N}^{sq}_{E\rightarrow E'}}_{E\rightarrow E'F'}\in\mathbb{V}_{sym}} I(S;E')_{\sigma^{opt}_{SE'}}\notag\\
    &=\mathcal{S}(\sigma_S^{opt})-\sup_{\mathcal{N}^{sq}_{E\rightarrow E'}\in\mathbb{N}_{sym}}I(S;E')_{\sigma^{opt}_{SE'}}
\end{align}
where $\mathbb{N}_{sym}$ is the set of symmetric channels. The last equality holds true because the mutual information 
$I(S;E')_{\sigma^{opt}_{SE'}}$ only depends on the squashing channel, not on its isometric extension. 
}}}

In the next coming subsections, we consider two specific cases. In the first one, we consider 50/50 beamsplitter for squashing channel and, in the second one we restrict the search for the optimal squashing channel to the set of symmetric qubit channels.}

\subsection{
{50/50 beamsplitter squashing channel}
}\label{subsec:onemode}
In this subsection, {we consider 50/50 beamsplitter as the squashing channel.}
{Among one-mode Gaussian symmetric channels the most well known one is 50/50 beamsplitter. Furthermore, this choice is in line with the} results in \cite{TGW2014}, where it is shown that in the set of pure-loss channels, 50/50 beamsplitter is the optimal squashing channel.

{A beamsplitter has two inputs, one playing the role of the environment and the other of the input to the channel.} When the environment mode is kept in the vacuum state, the beamsplitter performs as a Gaussian channel and is described by {the map} \cite{kraus_beam}:
\begin{equation}
    \label{eq:beam_splitter}
    \mathcal{N}^{BS}_\eta(\rho)=\sum_{k=0}^\infty B_k(\eta)\rho B_k^\dagger(\eta),
\end{equation}
where $\eta\in(0,1)$ is the transmissivity of the beamsplitter and $B_k(\eta)$s are {the Kraus operators taking the following explicit form in the Fock basis:}
\begin{equation}
    \label{eq:beamsplitter_kraus}
    B_k(\eta)=\sum_{m=0}^\infty\sqrt{m+k\choose k}(1-\eta^2)^\frac{k}{2}\eta^m\ket{m}\bra{m+k}.
\end{equation}
The beamsplitter transforms a single mode coherent {input} state $\ket{\beta}$ to a single mode coherent {output} state $\ket{\eta\beta}$, with a smaller amplitude \cite{kraus_beam}:
\begin{equation}
    \mathcal{N}^{BS}_\eta(\ket{\beta}\bra{\beta})=\ket{\eta\beta}\bra{\eta\beta}.
\end{equation}
 In this representation, a 50/50 beamsplitter corresponds to $\eta=\frac{1}{\sqrt{2}}$. Therefore, according to Eq.~(\ref{eq:OptInt}), an upper bound on the squashed entanglement of bosonic dephasing channel can be obtained by the following relation:
\begin{equation}\label{eq:g_upper}
    \tilde{E}_{sq}(\mathcal{N}^\gamma_{{S\rightarrow S}},a^\dagger a,N)\leq\sup_{p_n}\Big(\mathcal{S}(\sigma_S^{opt})-I(S;E')_{\sigma^{opt}_{SE'}}\Big),
\end{equation}
where $\sigma^{opt}_S$ and $\sigma^{opt}_{E'}$ are obtained by partial {tracing the following density operator with respect to $S$ and ${E'}$} 
\begin{align}\label{eq:opt_sigma}
    \sigma_{SE'}^{opt}&={{\sum_n}'} p_n\ket{n}_S\bra{n}\otimes\mathcal{N}^{BS}_\frac{1}{\sqrt{2}}\big(\ket{-\operatorname{i}\sqrt{\gamma}n}_{E'}\bra{-\operatorname{i}\sqrt{\gamma}n}\big)\cr
    &={{\sum_n}'} p_n\ket{n}_S\bra{n}\otimes\ket{-\frac{\operatorname{i}}{\sqrt{2}}\sqrt{\gamma}n}_{E'}\bra{-\frac{\operatorname{i}}{\sqrt{2}}\sqrt{\gamma}n}.
\end{align}
As $\sigma_{SE'}^{opt}$ is a separable state,
\begin{align}
   \mathcal{S}(\sigma_S^{opt})-I(S;E')_{\sigma_{SE'}^{opt}}&=\mathcal{S}(\sigma_{SE'}^{opt})-\mathcal{S}(\sigma_{E'}^{opt})\cr
   &=\mathcal{S}(\sigma_{S}^{opt})-\mathcal{S}(\sigma_{E'}^{opt}).
\end{align}
Therefore, Eq.~(\ref{eq:g_upper}) is simplified to
\begin{align}\label{eq:upper_gaussian_bound}
    \tilde{E}_{sq}(\mathcal{N}^\gamma_{S\rightarrow S},a^\dagger a, N)&\leq\sup_{p_n}\Big(\mathcal{S}(\sigma_{S}^{opt})-\mathcal{S}(\sigma_{E'}^{opt})\Big).
\end{align}
From Eq.~(\ref{eq:opt_sigma}) and Eq.~(\ref{eq:upper_gaussian_bound}) it is concluded that:
\begin{align}
    &\tilde{E}_{sq}(\mathcal{N}^\gamma_{S\rightarrow S},a^\dagger a,N)\leq\cr
    &\sup_{p_n}\Bigg(\mathcal{S}\Big({{\sum_n}'} p_n\ket{n}\bra{n}\Big)
    -\mathcal{S}\Big({{\sum_n}'} p_n\ket{\frac{n}{i}\sqrt{\frac{\gamma}{2}}}\bra{\frac{n}{i}\sqrt{\frac{\gamma}{2}}n}\Big)\Bigg)\notag\\
    \label{eq:ubound0}\\
    &=\sup_{p_n}\Bigg(\mathcal{S}\Big( {{\sum_n}'} p_n\ket{n}\bra{n}\Big)
    -\mathcal{S}\Big( {{\sum_n}'} p_n\ket{-\sqrt{\frac{\gamma}{2}}n}\bra{-\sqrt{\frac{\gamma}{2}}n}\Big)\Bigg).\notag\\
    \label{eq:uperbound_g}
\end{align}
{The equality \eqref{eq:uperbound_g}} is due to the invariance property of von-Neumann entropy under the unitary conjugation that transforms coherent state $\ket{\operatorname{i}\alpha}$ to $\ket{\alpha}$,  
$\forall\alpha\in\mathbb{C}$.
In \cite{qc_dephasing} it is shown that the {right hand side  (r.h.s.) of} Eq.~(\ref{eq:uperbound_g}) is
the quantum capacity of a bosonic dephasing channel with a dephasing parameter $\gamma'=\frac{\gamma}{2}$ {and the} energy constraint at the input, that is 
\begin{align}\label{eq:upperbound_sq_qc}
    \mathcal{Q}^{LOCC}_{S\leftrightarrow {S}}(\mathcal{N}_{{S\rightarrow S}}^\gamma,a^\dagger a,N)&\leq\tilde{E}_{sq}(\mathcal{N}_{{S\rightarrow S}}^\gamma,a^\dagger a,N)\cr
    &\leq  Q(\mathcal{N}_{S\rightarrow {S}}^{\frac{\gamma}{2}},a^\dagger a, N).
\end{align}
{Here $Q(\mathcal{N}_{S\rightarrow {S}}^{\gamma},a^\dagger a, N)$ denotes the quantum capacity of bosonic dephasing channel with dephasimg parameter $\gamma$ and average input energy $N$.}

So far, we have derived an upper bound for the energy-constrained {squashed entanglement of the channel}, which {in turn} is an upper bound for the energy-constrained two-way LOCC-assisted capacity.
{ 
Next we derive a lower bound for the energy-constrained two-way LOCC-assisted capacity.
}
In \cite{RCI}, a lower bound on the two-way LOCC-assisted quantum capacity was introduced with the name of reverse coherent information \cite{HorodeckiF2000, DevetakWinter2005, DJKR2006}. The reverse coherent information of a channel $\mathcal{N}_{S\rightarrow S'}$ is defined as
\begin{equation}
\label{eq:reverse}
    \mathcal{I}_R(\mathcal{N}_{S\rightarrow S'}):=\sup_{\rho_S}\Big(\mathcal{S}(\rho_S)-\mathcal{S}(\mathcal{N}^c_{S\rightarrow E}(\rho_S))\Big).
\end{equation}
{It is shown in \cite{reverse} that for a general channel $\mathcal{N}_{S\rightarrow S'}$ the following inequalities hold
\begin{equation}
\label{eq:ineqiRCI}
\mathcal{Q}(\mathcal{N}_{S\rightarrow S'},a^\dagger a,N)\leq \mathcal{I}_R(\mathcal{N}_{S\rightarrow S'})\leq\mathcal{Q}^{LOCC}_{S\leftrightarrow S'}(\mathcal{N}_{S\rightarrow S'},a^\dagger a,N)
\end{equation}
}
{For bosonic dephasing channel we know that the quantum capacity} is achieved by using a mixture of {Fock} states as input, which is invariant under the action of the channel, {namely}
\begin{equation}
    Q(\mathcal{N}_{S\rightarrow {S}}^\gamma)=\sup_{\rho'_S}\Big(\mathcal{S}(\rho'_S)-\mathcal{S}(\mathcal{N}^{\gamma^c}_{{S\rightarrow E}}(\rho'_S))\Big)
\end{equation}
{{where $\rho'_S$ belongs to the set of mixture of Fock states {\cite{qc_dephasing}}.}}
In appendix \ref{sec:app_1} we show that the quantum capacity of the bosonic dephasing channel and its reverse coherent information are equal:
\begin{equation}
    \label{eqq:equality of rev_and_qc}
        \mathcal{I}_R(\mathcal{N}^\gamma_{S\rightarrow {S}})=Q(\mathcal{N}_{S\rightarrow {S}}^\gamma)
\end{equation}
{As constraining the average input energy, within a bounded error leads to truncating the Hilbert-space dimension and the arguments supporting the equality in (100) are valid over truncated Hilbert space dimension\footnote{\label{footnote:energy_const}If we constraint the input average energy as
$\forall\epsilon>0, \; 
\exists d$, such that $\forall N>d : |\sum_{n=0}^{N}np_n-E|<\epsilon$ {then within a bounded error, it results}
truncating the input Hilbert space.},} 
we conclude that for a bosonic dephasing channel the lower bound on its energy-constrained two-way LOCC-assisted quantum capacity
$\mathcal{Q}^{LOCC}_{S\leftrightarrow S'}(\mathcal{N}^\gamma_{S\rightarrow S},a^\dagger a,N)$
is equal to its energy-constrained quantum capacity with parameter $\gamma$. Therefore, taking into
account Eq.~(\ref{eq:upperbound_sq_qc}), {(\ref{eq:ineqiRCI}) and (\ref{eqq:equality of rev_and_qc})} we arrive at
\begin{align}\label{eq:lower_upper}
   Q(\mathcal{N}_{S\rightarrow S'}^\gamma,a^\dagger a, N)&\leq \mathcal{Q}^{LOCC}_{S\leftrightarrow S'}(\mathcal{N}^\gamma,a^\dagger a,N)\cr
   &\leq  Q(\mathcal{N}_{S\rightarrow S'}^{\frac{\gamma}{2}},a^\dagger a, N).
\end{align}
With reference to \cite{qc_dephasing} we can compute both the lower bound and the upper bound in Eq.~(\ref{eq:lower_upper}). Fig.~(\ref{fig:gaussian_L_U_Bound}) represents these bounds for Hilbert spaces {truncated to the} dimension $d=3$ (red curves), $d=10$ (blue curves), versus the noise parameter $\gamma$. In Fig.~(\ref{fig:gaussian_L_U_Bound}) solid curves correspond to the upper bound in Eq.~(\ref{eq:lower_upper}) and dashed {curves} correspond to the lower bound in Eq.~(\ref{eq:lower_upper}).  

As {one can see in Fig.~(\ref{fig:gaussian_L_U_Bound}) lower bound and upper bounds are very close to each other, confirming their tightness.} To better illustrate this fact, in Fig.~(\ref{fig:gaussian_L_U_Bound_rel}) {it is shown the difference between upper and lower bounds} versus noise parameter $\gamma$ for Hilbert-space with dimension $d=2$ (red curves), $d=3$ (green curves), $d=9$ (blue curves), and $d=10$ (orange curves).
As expected, this difference {vanishes} at $\gamma=0$. Furthermore, {as it is seen in Fig.~(\ref{fig:gaussian_L_U_Bound_rel})}, it decreases 
 as well for large vales of the noise parameter.  
 \begin{figure}
    \centering
    \includegraphics[width=\columnwidth]{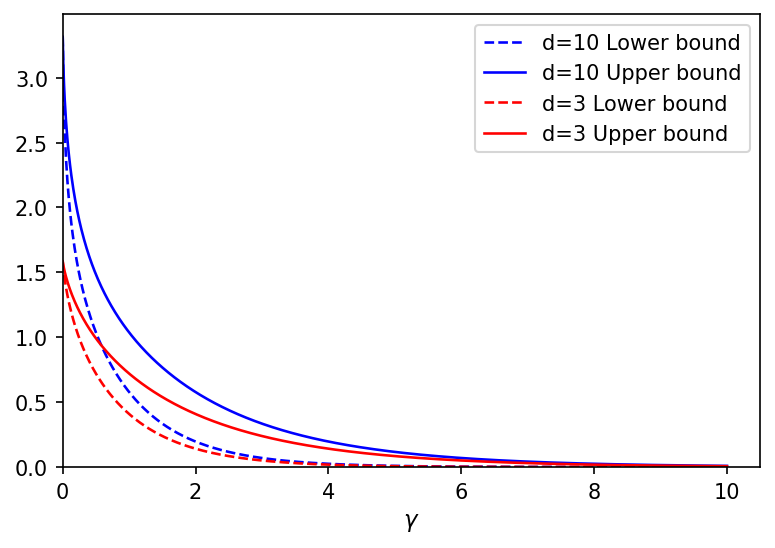}
    \caption{Upper bound and lower bound on  energy-constrained LOCC-assisted quantum capacity, $\mathcal{Q}^{LOCC}_{S\leftrightarrow S'}(\mathcal{N}^\gamma,a^\dagger a,N) $, as given in  Eq.~(\ref{eq:lower_upper}) versus noise parameter $\gamma$. Solid lines correspond to the upper bound, while dashed lines correspond to the lower bound. Different colors  refer to different dimensions of the truncated Hilbert space.}
    \label{fig:gaussian_L_U_Bound}
\end{figure}

\begin{figure}
    \centering
    \includegraphics[width=\columnwidth]{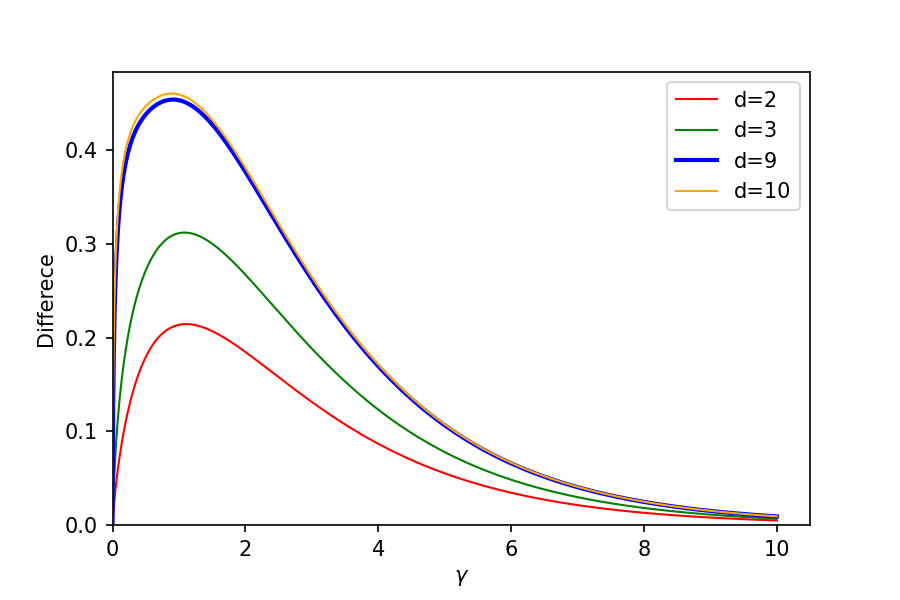}
    \caption{{Difference between upper} bound and lower bound on $\mathcal{Q}^{LOCC}_{S\leftrightarrow S'}(\mathcal{N}^\gamma,a^\dagger a,N) $ as given in Eq.~(\ref{eq:lower_upper}) versus noise parameter $\gamma$. Different curves correspond to different {dimensions of the truncated Hilbert space:} $d=2$ (red curve), {$d=3$} (green curve), $d=9$ (blue curve) and $d=10$ (yellow curve).}
    \label{fig:gaussian_L_U_Bound_rel}
\end{figure}
{As it is shown in Fig.~(\ref{fig:gaussian_L_U_Bound_rel}) for $\gamma<9$ the difference between the lower bound and upper bound in Eq.~(\ref{eq:lower_upper}) increases by increasing the dimension of Hilbert-space}. 
It is worth noticing that the upper bounds {corresponding to different truncated Hilbert-spaces, get close to each other for dimensions larger than $d=9$.} The same happens for the lower bounds. The saturation of upper and lower bounds with the {increasing} Hilbert space dimension can be seen in Fig.~(\ref{fig:uper_d}).
{This effect is in} agreement with the result in \cite{qc_dephasing}. As it is shown in \cite{qc_dephasing} the quantum capacity of a bosonic dephasing channel is equal to the quantum capacity of this channel in the truncated Hilbert space {of dimension 9.} Therefore, {we can conclude} that LOCC-assisted quantum capacity of bosonic dephasing channel without any input energy constraint is upper and lower bounded by {the quantum capacity in the following way:}
\begin{equation}\label{eq:lower_upper_un}
   Q(\mathcal{N}_{S\rightarrow S'}^\gamma) \leq \mathcal{Q}^{LOCC}_{S\leftrightarrow {S}}(\mathcal{N}^\gamma_{S\rightarrow S}) \leq  Q(\mathcal{N}_{S\rightarrow {S}}^{\frac{\gamma}{2}}).
\end{equation}
In conclusion, if {we use 50/50 beamsplitter as squashing channel for a bosonic dephasing channel}, we successfully obtain a lower and an upper bound for two-way LOCC-assisted quantum capacity of the bosonic dephasing channel, with energy constraint (Eq.~(\ref{eq:lower_upper})) and without energy constraint (Eq.~(\ref{eq:lower_upper_un})). As discussed above, these bounds are tight. {For another example of squashing channel}, in the next subsection, we analyze possible candidates among qubit channels. \begin{figure}
    \centering
    \includegraphics[width=\columnwidth]{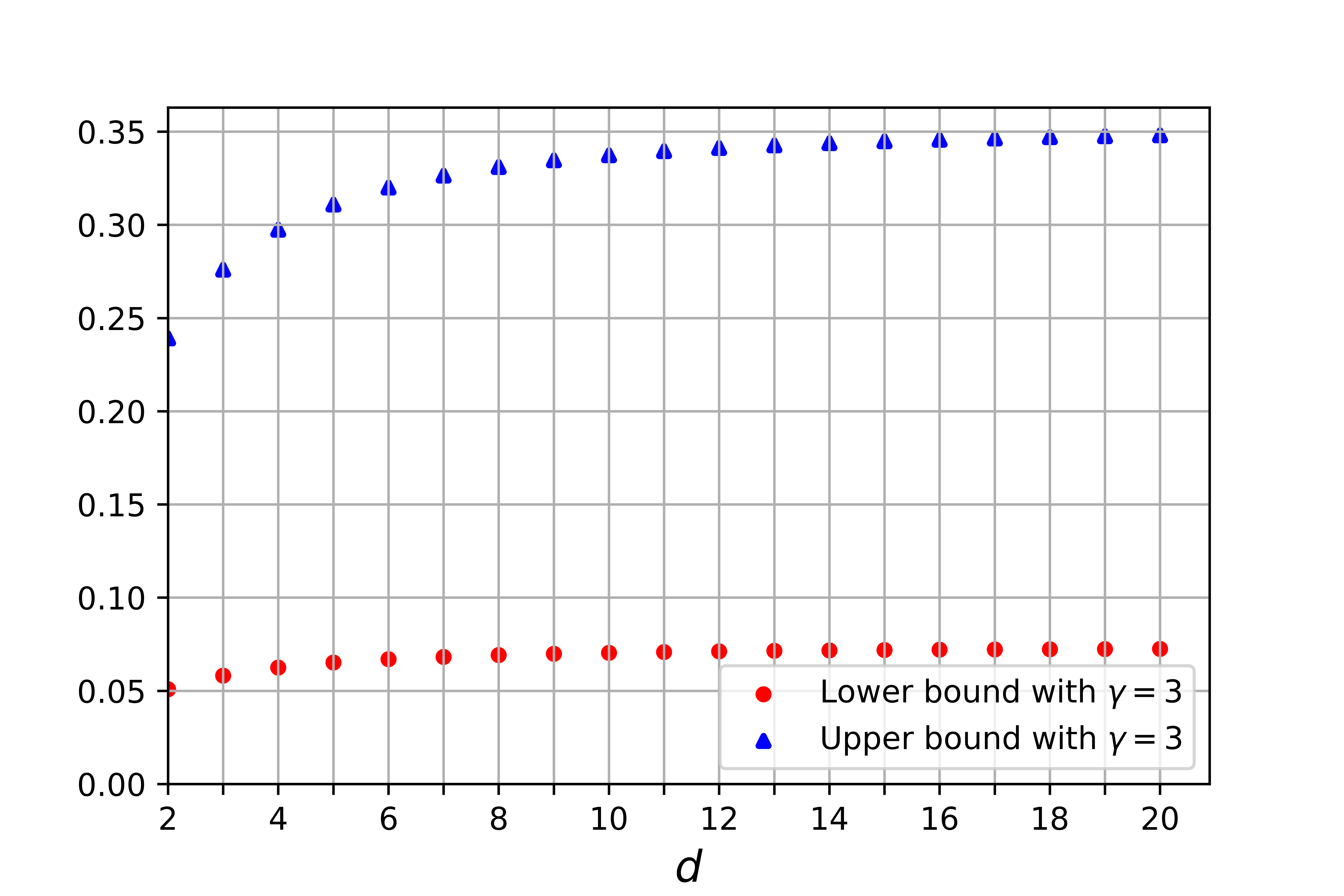}
    \caption{Upper bound (blue triangles) and lower bound (red circles) of  $\mathcal{Q}^{LOCC}_{S\leftrightarrow S'}(\mathcal{N}^\gamma_{S\rightarrow S},a^\dagger a,N) $ given in Eq.~(\ref{eq:lower_upper}) versus Hilbert-space dimension, $d$, when noise parameter $\gamma=3$.}
    \label{fig:uper_d}
\end{figure}


\subsection{Qubit squashing channels}\label{sbc:qubit_sq}
In this subsection, we truncate the infinite-dimensional Hilbert space to a two-dimensional Hilbert space and search for the best qubit squashing channel. 
Among the qubit channels the upper bound for LOCC-assisted quantum capacity of the generalized amplitude damping channel is analyzed by constructing particular squashing channels \cite{KSW2020}. 
Here, our focus is on bosonic dephasing channel in truncated two dimensional Hilbert space and our approach is to use the characterization of symmetric qubit channels \cite{self-comp} to find the one which maximized the mutual information in Eq.~(\ref{eq:self_comp_sq}).
\begin{figure}
    \centering
    \includegraphics[width=\columnwidth]{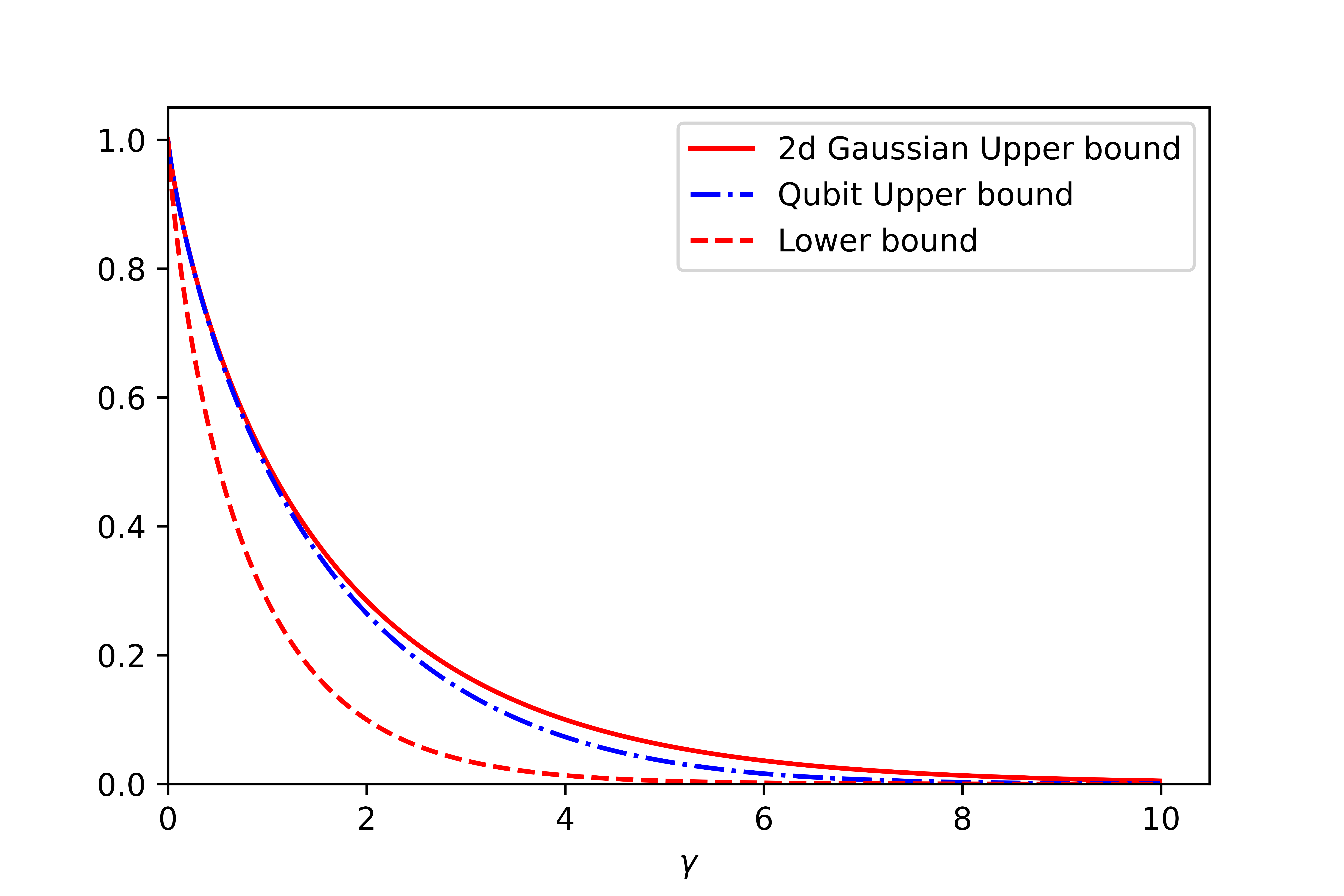}
    \caption{Comparison between upper bounds on  energy-constrained LOCC-assisted quantum capacity, $\mathcal{Q}^{LOCC}_{S\leftrightarrow S'}(\mathcal{N}^\gamma,a^\dagger a,N) $ obtained using one mode Gaussian symmetric channel (solid red curve) and qubit symmetric channel (dash-dotted blue curve) as functions of the dephasing parameter $\gamma$. Red dashed curve shows the lower bound on  energy-constrained LOCC-assisted quantum capacity, $\mathcal{Q}^{LOCC}_{S\leftrightarrow S'}(\mathcal{N}^\gamma,a^\dagger a,N) $ obtained by reverse coherent information.}
    \label{fig:non-g_g_L_U_Bound}
\end{figure}
Following Eq.~(\ref{eq:dephasing_output}), an optimal input state on the truncated input Hilbert space with dimension {two} has the following form:
\begin{equation}
\label{eq:two_dim_opt}
    \sigma_{SE'}^{opt}=\sum_{n=0}^1 p_n\ket{n}_S\bra{n}\otimes\mathcal{N}_{{E\rightarrow E'}}^{sq}(\ket{-\operatorname{i}\sqrt{\gamma}n}_E\bra{-\operatorname{i}\sqrt{\gamma}n}).
\end{equation}
Here $\mathcal{N}_{{E\rightarrow E'}}^{sq}$ is defined on bounded operators over infinite-dimensional Hilbert space. However, by truncating the input Hilbert space, the action of squashing channel $\mathcal{N}_{{E\rightarrow E'}}^{sq}$ is effectively restricted to bounded operators over Hilbert space spanned by {$\{\ket{0}, \ket{-\operatorname{i}\sqrt{\gamma}n}\}$. By employing the Gram-Schmidt procedure 
we construct orthonormal states as follows}
\begin{align}
&\ket{e_0}{:=}\ket{0},\cr
&\ket{e_1}{:=}\frac{\ket{-\operatorname{i}\sqrt{\gamma}}-\langle 0\ket{-\operatorname{i}\sqrt{\gamma}}\ket{e_0}}{\|\ \ket{-\operatorname{i}\sqrt{\gamma}}-\langle 0\ket{-\operatorname{i}\sqrt{\gamma}}\ket{e_0}\|\ }.
\end{align}
Furthermore, as we are restricting our attention to {symmetric} squashing channels, the input and output spaces of {squashing channel} are isomorphic, so we denote it by $\mathcal{N}_{{E\rightarrow E}}^{sq}$. Hence, the squashing channel in Eq.~(\ref{eq:two_dim_opt}) is effectively a qubit channel. Therefore, to derive an upper bound for {squashed entanglement of the channel} over two dimensional Hilbert space, we need to compute the right hand side of the following inequality from Eq.~(\ref{eq:self_comp_sq})
\begin{align}
\label{eq:qubit_uper}
    &\tilde{E}_{sq}(\mathcal{N}_{{S\rightarrow S}}^\gamma,a^\dagger a,N)\leq
    \cr
    &\sup_{p_n}\Bigg( \mathcal{S}(\sigma_S^{opt})-\sup_{\mathcal{N}^{sq}_{E\rightarrow {E}}\in\mathbb{N}_{sym}}I(S;E')_{\sigma^{opt}_{SE'}}\Bigg),
\end{align}
with $\sigma_{SE'}^{opt}$ {given} in Eq.~(\ref{eq:two_dim_opt}) and $\mathcal{N}_{{E\rightarrow E}}^{sq}$ being a {symmetric} qubit {channels characterized in \cite{self-comp}. These latter are described by either of the following sets of Kraus operators:}
\begin{align}\label{eq:2dself_1}
    K_1=\begin{pmatrix}
    \sin(\theta)     &  0\\ 
        0 & \frac{1}{\sqrt{2}}
    \end{pmatrix}\qquad\quad
    K_2=\begin{pmatrix}
        0 & \frac{1}{\sqrt{2}} \\
        e^{\operatorname{i}\phi}\cos(\theta) & 0
    \end{pmatrix}.
\end{align}
and
\begin{align}\label{eq:2dself_2}
    K'_1=\begin{pmatrix}
    1    &  0\\ 
        0 & \frac{1}{\sqrt{2}}\sin(\theta) 
    \end{pmatrix}\qquad\quad 
    K'_2=\begin{pmatrix}
        0 & \frac{1}{\sqrt{2}}\sin(\theta)  \\
        0 & e^{\operatorname{i}\phi}\cos(\theta)
    \end{pmatrix}.
\end{align}
In both cases $\theta\in [0,\pi]$ and $\phi\in [0,2\pi]$. In principle, all the terms on the right hand side of Eq.~(\ref{eq:qubit_uper}) can be computed analytically. However, the final expression after doing the required diagonalization for computing different terms is complicated, and optimization over such expressions is essential.
Hence, we do the optimization on the right hand side of Eq.~(\ref{eq:qubit_uper}) numerically. 

First, we did the optimization over all qubit {symmetric} channels and then {found} the maximum over all probability distributions. The outcome of our numerical analysis for the upper bound is depicted in Fig.~(\ref{fig:non-g_g_L_U_Bound}) with the dashed-dotted blue curve. For better comparison in Fig.~(\ref{fig:non-g_g_L_U_Bound}) we have also presented the upper bound (solid red line) and the lower bound (dashed red line) when the truncated input Hilbert space is two-dimensional, and the squashing channel is a 
Gaussian channel as discussed in \S\ref{subsec:onemode}. Hence we conclude that, at least for two-dimensional truncated input Hilbert space, {symmetric} qubit channels outperform one-mode Gaussian squashing channels for intermediate values of {$2<\gamma<8$}. {However}, the difference {between the upper bounds given by 50/50 beamsplitter and symmetric qubit channels } is negligible for {$2<\gamma<8$}, and vanishes for {the rest of the values
of $\gamma$}.


\section{Conclusion}\label{sec:conc}

We have analyzed the LOCC-assisted quantum capacity {of bosonic} dephasing channel subject to energy constraints on input states. 
Quantum capacity of this non-Gaussian channel is derived in \cite{qc_dephasing} and our aim 
here was to address the role of LOCC on the reliable rate of entanglement establishment between sender and receiver.

Although for practical reasons it is essential to know the LOCC-assisted quantum capacity of the channels \cite{TGW2014}, there is no compact formula in terms of entropic functions for quantifying it. It  was proved that {squashed entanglement of a channel} is an upper bound for LOCC-assisted quantum capacity \cite{wildelocc} and secret-key agreement capacity \cite{TGW2014}. In order to compute {squashed entanglement of a channel}, optimization over {both the set of input states and the} set of isometric extensions of quantum channels for the optimal squashing channel is required. {Thus, 
even computing a bound (for the LOCC-assisted quantum capacity) through such an optimization results in general} is very challenging, if not impossible. 

To overcome these complications for computing the upper bound for LOCC-assisted quantum capacity {of bosonic} dephasing channel, first, we {used} the phase covariant property of the channel to determine the structure of the optimal input state. In other words, we prove that it is sufficient to search for the optimal input state over the set of diagonal states in the {Fock} basis instead of the whole set of density operators. For the optimization over {the} isometric extensions of the  {squashing} channels,
{we prove that it is sufficient to restrict the search to the set of SQMCI channels of the channel environment output, if this set is non-empty.This result is valid for a generic given channel. However for the bosonic dephasing channel we found arguments supporting the conjecture that the set of SQMCI squashing channels is null. Hence, we choose the explicit examples of squashing channels for bosonic dephasing channels from the set of symmetric channels.}

{For the explicit examples, first we considered a 50/50 beamsplitter for the squashing channel. We showed that in this case}, the LOCC-assisted quantum capacity of the channel is less than or equal to the quantum capacity of a bosonic dephasing channel {having the noise parameter} reduced by a factor two. Furthermore, to {derive the tightness possible lower bound}, we {used} the result in \cite{RCI} where a lower bound on LOCC-assisted quantum capacity is introduced in terms of reverse coherent information. 
{We proved} that the reverse coherent information and the quantum capacity of bosonic dephasing channels are equal. Hence, when LOCC is allowed, the reliable rate for {sharing} entanglement between the two parties increases. Therefore, taking into account the results in \cite{qc_dephasing}, we {provided computable upper and lower bounds} for LOCC-assisted quantum capacity of bosonic dephasing channel, and we showed that this result is valid whether or not the input state is subject to energy constraints. More importantly, we showed that these bounds are tight, {meaning} that the quantum capacity of a bosonic dephasing channel with noise parameter $\gamma$ when LOCC is allowed is very close to the quantum capacity of a bosonic dephasing channel with {noise} parameter $\frac{\gamma}{2}$. 
{In other words, with the assistance of LOCC the effective noise parameter is halved.}
 
To extend our analysis beyond {50/50 beamsplitter squashing channel}, we also {discussed} the case where the squashing channel is a {symmetric} qubit channel. As {put forward} in \S\ref{sbc:qubit_sq}, although the upper bound given by the optimal qubit squashing channel is smaller than the one with the optimal one-mode Gaussian channel for {some range of noise parameter}, their difference is negligible.

Our results not only set as an explicit example
to confirm the importance and tightness of the upper {bound} in terms of {squashed entanglement of the channel}, but also motivates analysing the quantum capacity of other non-Gaussian channels {especially} with the assistance of LOCC. {Additionally, it seems interesting to devote further investigations} on characterizing the set of {QMCI and SQMCI channels for particular classis of initial states} and addressing their performance as squashing channels.


\section{Appendix}
\subsection{ On the equality of quantum capacity and reverse coherent information {for} bosonic dephasing channel}
\label{sec:app_1}

As defined in Eq.~(\ref{eq:reverse}), the reverse coherent information {of bosonic} dephasing channel is given by
\begin{equation}\label{eq:reverse_app}
    \mathcal{I}_R(\mathcal{N}^\gamma_{S\rightarrow {S}})=\sup_{\rho_S}\Big(\mathcal{S}(\rho_S)-\mathcal{S}(\mathcal{N}^{\gamma^c}_{{S\rightarrow E}}(\rho_S))\Big).
\end{equation}
Moreover, its quantum capacity is proven to be given by the following optimization problem where the supremum is taken over a mixture of {Fock} states $\rho'_S$ \cite{qc_dephasing}:
\begin{equation}\label{eq:qc_app}
    Q(\mathcal{N}_{S\rightarrow {S}}^\gamma)=\sup_{\rho'_S}\Big(\mathcal{S}(\rho'_S)-\mathcal{S}(\mathcal{N}_{S\rightarrow E}^{\gamma^c}(\rho'_S))\Big).
\end{equation}
Here we show that the supremum in Eq.~(\ref{eq:reverse_app}), as in the case of quantum capacity, is achieved by using mixture of {Fock} states. Therefore, we prove that
\begin{equation}\label{eq:equality of rev_and_qc}
    \mathcal{I}_R(\mathcal{N}^\gamma_{S\rightarrow {S}})=Q(\mathcal{N}_{S\rightarrow {S}}^\gamma).
\end{equation}
{To this end, we first define} the following function:
\begin{equation}
    I_R(\mathcal{N}_{S\rightarrow {S}},\rho_S)\coloneqq\mathcal{S}(\rho_S)-\mathcal{S}(\mathcal{N}^c_{S\rightarrow E}(\rho_S)).
\end{equation}
{Next, we prove} that $I_R(\mathcal{N}_{S\rightarrow {S}},\rho_S)$ is concave with respect to input states $\rho_S$. Consider the following two classical-quantum states:
\begin{align}\label{eq:sigma_tau}
    \sigma_{XS}&=\sum_n p_n\ket{n}_X\bra{n}\otimes\rho_S^{(n)},\cr
    \tau_{XE}&=\sum_n p_n\ket{n}_X\bra{n}\otimes\mathcal{N}^c_{S\rightarrow E}(\rho_S^{(n)}),
\end{align}
where $\mathcal{N}^c_{S\rightarrow E}$ is {the complementary channel to the} channel $\mathcal{N}_{S\rightarrow {S}}$. The following inequalities hold true:
\begin{align}\label{ineq:hirearchy}
    I(X;E)_{\tau_{XE}}&\leq I(X;S)_{\sigma_{XS}},\cr
    {\mathcal{S}}(\tau_{E})-{\mathcal{S}}(E|X)_{\tau_{XE}}&\leq \mathcal{S}(\sigma_{S})
    -{\mathcal{S}}(S|X)_{\sigma_{XS}},\cr
    \mathcal{S}(S|X)_{\sigma_{XS}}-\mathcal{S}(E|X)_{\tau_{XE}}&\leq \mathcal{S}(\sigma_{S})- \mathcal{S}(\tau_{E}).\cr
\end{align}
The first inequality is due to the data processing inequality of mutual information, the second inequality follows the definition of mutual information, and the third inequality is just a rearrangement. Therefore, due to the classical-quantum nature of the states $\sigma_{XS},\tau_{XE}$ in Eq.~(\ref{eq:sigma_tau}), and the last inequality of \eqref{ineq:hirearchy}, we have:
\begin{align}
    \sum_n p_n\Big(&\mathcal{S}(\rho_S^{(n)})-\mathcal{S}(\mathcal{N}^c_{S\rightarrow E}(\rho_S^{(n)}))\Big)\cr
    &\leq \mathcal{S}(\sum_n p_n\rho_S^{(n)})-\mathcal{S}(\sum_n p_n\mathcal{N}^c_{S\rightarrow E}(\rho_S^{(n)})),
    \cr
\end{align}
which is equivalent to
\begin{equation}\label{eq:conc_rev}
    \sum_n p_n I_R(\mathcal{N}_{S\rightarrow {S}},\rho_S^{(n)})\leq I_R\Big(\mathcal{N}_{S\rightarrow {S}},\sum_n p_n\rho_S^{(n)}\Big),
\end{equation}
for any probability distribution $p_n$. {As a consequence} $I_R(\mathcal{N}_{S\rightarrow {S}},\rho_S)$ is a concave function with respect to its {argument} $\rho_S$.

On the other hand, according to Eq.~(\ref{eq:comp_inv}) the complementary channel of a bosonic dephasing channel, is invariant under the phase shift operator {of Eq.\eqref{Eq:PhaseO}.}
{Then, by} considering the unitarily invariance property of the von-Neumann entropy along with the invariance property of the complementary channel {of bosonic} dephasing channel under the action of $U_\theta$, we conclude that:
\begin{equation}\label{eq:unit_inv_reverse}
    I_R(\mathcal{N}^\gamma_{S\rightarrow {S}},{\rho_S({\theta})})=I_R(\mathcal{N}^\gamma_{S\rightarrow {S}},\rho_S)
\end{equation}
where ${ \rho_S({\theta)}:=}U_\theta\rho_S U_\theta$. Employing the concavity of $I_R(\mathcal{N}_{S\rightarrow S'},\rho_S)$ as in Eq.~(\ref{eq:conc_rev}), and Eq.~(\ref{eq:unit_inv_reverse}) {we are led to:}
\begin{equation}\label{eq:conc+inv}
   I_R(\mathcal{N}^\gamma_{S\rightarrow {S}},\rho_S)\leq I_R\Big(\mathcal{N}^\gamma_{S\rightarrow {S}},\int_0^{2\pi}d\theta p(\theta){\rho_S({\theta})}\Big).
\end{equation}
Taking $p(\theta)$ as a flat distribution and expanding $\rho_S$ in the Fock basis, $\rho_S=\sum_{m,n}\rho_{m,n}\ket{m}_S\bra{n}$, we have:
\begin{equation}\label{eq:flat_dist}
    \int_0^{2\pi}{\rho_S({\theta})}p(\theta)d\theta=\frac{1}{2\pi}\sum_{m,n}\int_0^{2\pi}d\theta\rho_{m,n}e^{\operatorname{i}\theta(n-m)}\ket{n}_S\bra{m}.
\end{equation}
Inserting the above result into {the r.h.s.} of \eqref{eq:conc+inv} we end up with:
\begin{equation}
    I_R(\mathcal{N}^\gamma_{S\rightarrow {S}},\rho_S)\leq I_R\Big(\mathcal{N}^\gamma_{S\rightarrow {S}},\sum_n \rho_{n,n}\ket{n}_S\bra{n}\Big).
\end{equation}
Hence, the supremum in Eq.~(\ref{eq:reverse_app}) is achieved by a mixture of {Fock} states $\rho'_S$. In other words, the optimization in Eqs.~(\ref{eq:reverse_app}) and (\ref{eq:qc_app}) are over the same space and it proves their equality as expressed in Eq.~(\ref{eq:equality of rev_and_qc}).
\section*{ACKNOWLEDGMENTS}
A. A and L. M. {acknowledge} financial support by Sharif University of Technology, Office of Vice President for Research
under Grant No. G930209. L. M acknowledges the hospitality by the Abdus Salam International Centre for Theoretical Physics
(ICTP) where parts of this work were completed. 
S. M acknowledges
the funding from the European Union's
Horizon 2020 research and innovation program, 
under grant agreement QUARTET No 862644. The authors are grateful to Mark Wilde for a careful reading of the manuscript. {A. A would like to thank Ernest Y.-Z. Tan for enlightening discussion.}
\bibliographystyle{IEEEtran}
\bibliography{TLOCCAQC-IEEE.bib}
\end{document}